\documentclass[journal,twoside]{IEEEtran}
\usepackage{amssymb, amsfonts, mathrsfs, graphicx, cite, xcolor, amsthm, bigints, balance, multirow}
\newtheorem{theorem}{Theorem}

\newcounter{MYtempeqncnt}
\let\bs\boldsymbol
\allowdisplaybreaks
\begin{document}
\title{On the Performance of Downlink NOMA in Underlay Spectrum Sharing}
\author{Vaibhav~Kumar,~\IEEEmembership{Member,~IEEE,} Zhiguo~Ding,~\IEEEmembership{Fellow,~IEEE,} and~Mark~F.~Flanagan,~\IEEEmembership{Senior Member,~IEEE}
\thanks{V. Kumar and M. F. Flanagan are with School of Electrical and Electronic Engineering, University College Dublin, Dublin D04 V1W8, Ireland (e-mail: vaibhav.kumar@ieee.org, mark.flanagan@ieee.org).\\
Z. Ding is with School of Electrical and Electronic Engineering, The University of Manchester, Manchester M13 9PL, U.K. (e-mail: zhiguo.ding@manchester.ac.uk).}
\thanks{
This work was supported by Science Foundation Ireland (SFI) through Grant 17/US/3445 and co-funded through the European Regional Development Fund under Grant 13/RC/2077.}}

\markboth{Accepted in IEEE Transactions on Vehicular Technology 2021}%
{Kumar \MakeLowercase{\textit{et al.}}: On the Performance of Downlink NOMA in Underlay Spectrum Sharing}

\maketitle
\begin{abstract}
Non-orthogonal multiple access (NOMA) and spectrum sharing are two potential technologies for providing massive connectivity in beyond fifth-generation (B5G) networks. In this paper, we present the performance analysis of a multi-antenna-assisted two-user downlink NOMA system in an underlay spectrum sharing system. We derive closed-form expressions for the average achievable sum-rate and outage probability of the secondary network under a peak interference constraint and/or peak power constraint, depending on the availability of channel state information (CSI) of the interference link between secondary transmitter (ST) and primary receiver (PR). For the case where the ST has a fixed power budget, we show that  performance can be divided into two specific regimes, where either the interference constraint or the power constraint primarily dictates the performance. Our results confirm that the NOMA-based underlay spectrum sharing system significantly outperforms its orthogonal multiple access (OMA) based counterpart, by achieving higher average sum-rate and lower outage probability. We also show the effect of information loss at the ST in terms of CSI of the link between the ST and PR on the system performance. Moreover, we also present closed-form expressions for the optimal power allocation coefficient that minimizes the outage probability of the NOMA system for the special case where the secondary users are each equipped with a single antenna. A close agreement between the simulation and analytical results confirms the correctness of the presented analysis. 
\end{abstract}

\IEEEpeerreviewmaketitle
\section{Introduction} \label{Sec_Intro}
\IEEEPARstart{T}{he} commercial deployment of the 5G wireless communications network has already begun in mid-2019 in many countries. The first phase of 5G mobile communications is expected to be operating in the 3.6 GHz range. However, the amount of spectrum in the sub-GHz and below 6~GHz range, which will support many crucial applications in 5G, is very congested~\cite{Nekovee}. With the advent of many new wireless communication applications and services, the number of devices/users accessing the wireless spectrum is increasing very rapidly, resulting in the problem of \emph{spectrum scarcity}. On the other hand, it is well-known that the 3.5 GHz band (along with some ISM and mmWave bands) are currently under-utilized, and therefore \emph{spectrum sharing} is considered as a potential solution to enhance the spectrum usage efficiency~\cite{Gridlock, IntelRadio, ML_CR}. On the other hand, NOMA has also gained tremendous attention as a potential multiple access technique for the next-generation mobile communications network, as it can provide \emph{massive connectivity} and can also enhance the spectral efficiency~\cite{NOMA_book, Myths, mmWave}.

In general, spectrum sharing between a licensed/primary network and an unlicensed/secondary network can be accomplished in three ways: \emph{underlay}, \emph{interweave} and \emph{overlay}~\cite{Gridlock}. In the case of underlay spectrum sharing, the ST transmits simultaneously with the primary-user transmitter (PT) using the band of frequencies originally owned by the primary network, in such a manner that the interference inflicted by the secondary network on the primary network is below a tolerance limit. In interweave spectrum sharing, a cognitive engine first determines the spectrum bands for which the usage license is owned by a primary network and the secondary network uses those licensed bands when primary activity is not detected in those bands. Determination of these \emph{spectrum holes} by the cognitive engine is termed as \emph{spectrum sensing}. In the case of overlay spectrum sharing, the secondary user transmits simultaneously with the primary user, but compensates for the interference caused on the primary network by relaying a part of the primary user's message to the intended receiver(s). The fusion of NOMA and spectrum sharing has gained particular attention in the past few years, as it has the potential to provide massive connectivity and to further enhance the spectrum utilization efficiency in beyond-5G systems. 

For the case of overlay spectrum sharing, many notable works analyzing the achievable rate, outage probability, throughput and optimal power allocation have been presented for different NOMA systems such as multi-user secondary network, energy harvesting STs, relay-based cooperative systems and hybrid satellite-terrestrial networks~\cite{Nakagami, Kader_ICNC, NovelSpectrumSharingDing, KaderOverlayElsevier, Satellite}.

On the other hand, there has also been particular research attention given to underlay spectrum sharing NOMA systems. Different cooperative and non-cooperative NOMA-based spectrum sharing architectures were proposed in~\cite{NOMA-SS-Mag}, including underlay, overlay and cognitive radio (CR) inspired NOMA. It was shown in~\cite{NOMA-SS-Mag} that NOMA-based spectrum sharing outperforms its OMA-based counterpart in terms of outage probability. The outage probability analysis of a NOMA-based underlay spectrum sharing system was presented in~\cite{LargeScale}, where the power transmitted by the ST was constrained by a peak tolerable interference power at PRs as well as a peak power budget at the ST. In~\cite{Detect_and_Forward}, the outage performance analysis of a relay-based underlay spectrum sharing NOMA system, consisting of one ST, one detect-and-forward relay, one PR and two secondary receivers (SRs), was presented, where the power transmitted from the ST was constrained by a peak interference constraint at the PR as well as a peak power budget at the ST. However, in~\cite{Detect_and_Forward}, it was assumed that the transmission from the relay does not cause any interference at the PR (due to a large separation between them), and the signal received at the relay and the two SRs were also assumed to be free from any interference from the primary network. The analysis of outage probability for a relay-based spectrum sharing NOMA system considering the relay-to-PR interference was presented in~\cite{ImperfectCSI}. The outage probability and throughput analysis of an underlay spectrum sharing hybrid OMA/NOMA system consisting of a PT, a PR, an ST and two SRs was presented in~\cite{Hybrid}, where the authors considered both primary-to-secondary and secondary-to-primary interference. The power transmitted from the ST was constrained by a peak interference constraint at the PR as well as a peak power budget constraint at the ST. However, it is noteworthy that the closed-form expressions for the system throughput (or the average achievable sum-rate) was not derived in~\cite{Hybrid}. The performance analysis in terms of average achievable sum-rate, outage probability and asymptotic behavior (of outage probability) of a NOMA-based cooperative relaying system in an underlay spectrum sharing scenario, considering only the peak interference constraint, was presented in~\cite{Vaibhav} (here the authors assumed that the ST and relay do not have any power budget constraints). In~\cite{Amplify_and_Forward}, the analysis of the outage probability for an underlay spectrum-sharing-inspired amplify-and-forward relay-based two-user downlink NOMA system was presented, where the power transmitted from the ST was assumed to be constrained by a peak power budget at the ST as well as a peak interference constraint at the PR. 

In summary, for the case of NOMA-based underlay spectrum sharing, most existing research deals only with the outage probability analysis (as in~\cite{LargeScale, Detect_and_Forward, ImperfectCSI, Hybrid, Amplify_and_Forward}) or considers only the peak interference constraint at the PR (as in~\cite{Vaibhav}). It is also noteworthy that only single-antenna receivers were considered in~\cite{LargeScale, Detect_and_Forward, ImperfectCSI, Hybrid, Amplify_and_Forward}. Also note that perfect instantaneous CSI regarding the ST-PR link(s) was assumed to be available at the ST in~\cite{LargeScale, Detect_and_Forward, Vaibhav} and~\cite{ Amplify_and_Forward}, whereas imperfect instantaneous CSI regarding the ST-PR link was assumed to be available at the ST in~\cite{ImperfectCSI} and~\cite{Hybrid}. Motivated by this, in this paper, we present the average achievable sum-rate and outage probability analysis of a two-user downlink NOMA system in underlay spectrum sharing (over Rayleigh fading wireless channels) where both of the (secondary) users are assumed to be equipped with multiple antennas. We also consider that only statistical channel state information (CSI) is available at the ST regarding the links between the ST and the users\footnote{For the case of a spectrum sharing system, the ST needs to coordinate with the primary network in order to acquire CSI for the ST-PR link(s). This results in extra overhead at the ST as compared to the traditional (licensed) systems. Since for the secondary channels (i.e., the channels between the ST and the NOMA users), acquiring statistical CSI requires a smaller overhead as compared to that for instantaneous CSI, we have considered the former in order to reduce the communication overhead at the ST. However, it can be shown that it is straightforward to extend the analysis presented in this paper to the case when instantaneous CSI is available at the ST regarding the secondary channels.}, whereas, for the case of the link between the ST and PR, we consider different scenarios, as explained in~Table~\ref{TableConditions}. Hereafter, we will refer to the CSI of the ST-PR link as the interference-link CSI (IL-CSI).

\renewcommand{\arraystretch}{1.3}
\begin{table*}
\centering
\caption{Details of the different system configurations analyzed in this work.}
\label{TableConditions}
\begin{tabular}{|l|l|l|l|}
\hline 
\multirow{2}{*}{\textbf{Case}} & \multicolumn{1}{l|}{\textbf{Power budget} } & \multirow{2}{*}{\textbf{IL-CSI at ST}} & \multirow{2}{*}{\textbf{Insights (for NOMA)}}\tabularnewline
 & \textbf{at ST} &  & \tabularnewline
\hline 
\multirow{3}{*}{IntICSI} & \multirow{3}{*}{Unlimited} & \multirow{3}{*}{Instantaneous} & The difference between the outage probability of IntICSI and IntSCSI
becomes negligible for a large\tabularnewline
 &  &  & number of antennas and a large value of the peak tolerable interference.
The rate of decay of the \tabularnewline
 &  &  & outage probability depends on the minimum number of antennas at the
NOMA users. However, a \tabularnewline
\cline{1-3} \cline{2-3} \cline{3-3} 
\multirow{3}{*}{IntSCSI} & \multirow{3}{*}{Unlimited} & \multirow{3}{*}{Statistical} & significant difference can be noted between the achievable sum-rate
of the two systems, even for a \tabularnewline
 &  &  & large number of antennas and/or for a large value of interference.
Using asymptotic rate analysis, we \tabularnewline
 &  &  & show that there is no benefit (in terms of achievable rate) from having
multiple antennas at the far user.\tabularnewline
\hline 
\multirow{2}{*}{PowIntICSI} & \multirow{2}{*}{Limited} & \multirow{2}{*}{Instantaneous} & The system experiences an outage floor, as well as a rate saturation,
in the power-constrained regime. \tabularnewline
 &  &  & A higher power budget is advantageous for these systems in the power-constrained
regime. However,\tabularnewline
\cline{1-3} \cline{2-3} \cline{3-3} 
\multirow{2}{*}{PowIntSCSI} & \multirow{2}{*}{Limited} & \multirow{2}{*}{Statistical} & no significant benefit is observed from a higher power budget in the
interference constrained regime\tabularnewline
 &  &  & for PowIntICSI and PowIntSCSI system, whereas for the case of the
PowIntOneBit system, a higher \tabularnewline
\cline{1-3} \cline{2-3} \cline{3-3} 
\multirow{2}{*}{PowIntOneBit} & \multirow{2}{*}{Limited} & \multirow{2}{*}{No CSI} & power budget results in an inferior performance in the interference
constrained regime. The benefit\tabularnewline
 &  &  & of using a large number of antennas is significant in both the regimes. \tabularnewline
\hline 
\end{tabular}
\end{table*}

The main contributions of this paper are summarized as follows:
\begin{itemize}
	\item We derive closed-form expressions for the average achievable sum-rate and outage probability for the spectrum sharing NOMA system for all the five cases shown in~Table~\ref{TableConditions}. 
	\item For the special case where the secondary users are each equipped with a single receive antenna, we derive an explicit analytical expression for the optimal power allocation coefficient that minimizes the outage probability of the spectrum sharing NOMA system (except for the case of PowIntICSI). For the general case, where the users are equipped with more than one receive antenna, the value of optimal power allocation coefficient is obtained numerically.
	\item By comparing the performance of the spectrum sharing NOMA system with the corresponding OMA system, we show that the NOMA system outperforms its OMA-based counterpart by achieving lower outage probability and higher achievable rate. More interestingly, we show a performance comparison among different system configurations of the spectrum sharing NOMA system (as described in Table~\ref{TableConditions}) to show the effect of loss of information (in terms of CSI) on the overall system performance.
	\item For the case of NOMA systems, our results confirm that when the value of the peak tolerable interference at the PR is large and the SRs are equipped with a large number of antennas, the difference between the performance of IntICSI and IntSCSI is negligible in terms of outage probability, whereas the performance difference is significant in terms of achievable sum-rate. A similar trend is observed in the \emph{interference-constrained regime} (to be explained later) when the performance (as regards outage probability as well as achievable sum-rate) of PowIntICSI, PowIntSCSI and PowIntOneBit are compared. Interestingly, the results also confirm that for the case of PowIntOneBit NOMA, a higher power budget at the ST results in an inferior performance (both in terms of outage probability and achievable sum-rate) in the interference-constrained regime.
\end{itemize}
The achievable rate analysis of an underlay spectrum sharing OMA system (with a single downlink secondary user) considering peak or average interference constraint and unlimited or limited power-budget constraint, was presented in~\cite{Alouini}. The authors demonstrated the effect of different levels of IL-CSI on the achievable rate of the secondary user. The main difference between the work presented in~\cite{Alouini} and that of this paper is that we consider a (downlink) multiple access system, whereas in~\cite{Alouini} only a single downlink user was considered. Also, in the work presented in~\cite{Alouini}, the authors considered a single-antenna-assisted downlink user, while here we consider multi-antenna-assisted downlink users. It is also important to note that along with the achievable rate analysis, we also present the outage probability analysis of the spectrum sharing system, which was not presented in~\cite{Alouini}.

\section{System Model} \label{Sec_SysMod}
\begin{figure}[t]
	\centering
	\includegraphics[width = 0.7\linewidth]{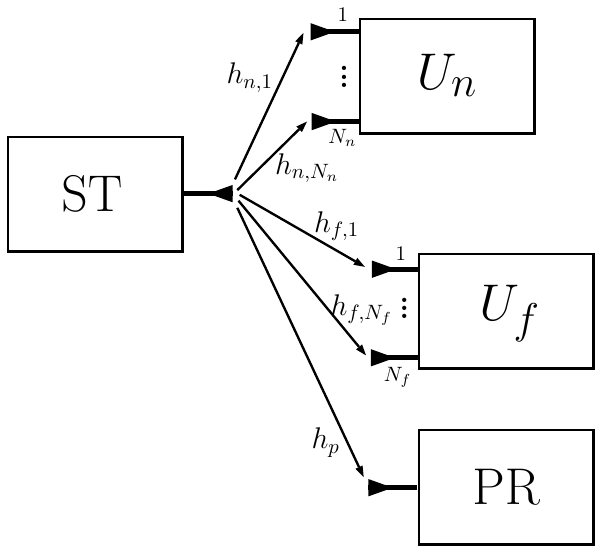}
	\caption{System model for underlay spectrum sharing. Here $U_n$ and $U_f$ are the secondary-user receivers.}
	\label{Fig_SysMod}
\end{figure}
Consider the system shown in Fig.~\ref{Fig_SysMod}, consisting of a secondary-user transmitter ST, a primary-user receiver PR and two secondary-user receivers $U_n$ and $U_f$. It is assumed that the ST and PR are each equipped with a single antenna, whereas $U_n$ and $U_f$ are equipped with $N_n (\geq 1)$ and $N_f (\geq 1)$ antennas, respectively. The channel fading coefficient between the ST and the PR is denoted by $h_p$, whereas that between the ST and the $i$-th antenna of $U_n$, and the ST and the $j$-th antenna of $U_f$ are denoted by $h_{n, i}$ and $h_{f, j}$, respectively, where $i \in \{1, 2, \ldots, N_n\}$ and $j \in \{1, 2, \ldots, N_f\}$. We assume that $h_p \sim \mathcal {CN}(0, \Omega_p = d_{p}^{-\alpha})$, $h_{n, i} \sim \mathcal {CN}(0, \Omega_n = d_n^{-\alpha})$ and $h_{f, j} \sim \mathcal {CN}(0, \Omega_f = d_f^{-\alpha})$ where $d_p$, $d_n$ and $d_f$ denote the distance between the ST and PR, ST and $U_n$, and ST and $U_f$, respectively, and $\alpha$ denotes the path-loss exponent. Throughout this paper, we assume that the ST has statistical channel state information (CSI) regarding the links between the ST and $U_n$, and ST and $U_f$, i.e., the knowledge of $\Omega_n$, $\Omega_f$ and the corresponding distribution of these links, whereas the availability of the CSI regarding the ST-PR link, i.e., the IL-CSI for different scenarios, is given in~Table~\ref{TableConditions}. It is also assumed that $d_n < d_f < d_p$, and we therefore refer to $U_n$ and $U_f$ as the near and far users respectively. We consider the scenario where the secondary network (consisting of ST, $U_n$ and $U_f$) is deployed outside the coverage of the primary-user transmitter, and therefore, we do not consider the interference at $U_n$ and $U_f$ caused by the primary-user transmitter. 

In the case of NOMA, the ST broadcasts a power-division multiplexed symbol 
$$\sqrt{a_n P_t(g_p)} x_n + \sqrt{a_f P_t(g_p)} x_f,$$ 
where $x_n$ and $x_f$ are the unit-energy symbols intended for users $U_n$ and $U_f$, respectively, $a_n$ and $a_f$ are the power allocation coefficients for users $U_n$ and $U_f$, respectively (we assume that $a_n < a_f$ and $a_n + a_f = 1$), $g_p \triangleq |h_p|^2$, and $P_t(g_p)$ is the total power transmitted from the ST. In general, $P_t(g_p)$ is a one-to-one mapping from the channel gain $g_p$ to the set of non-negative real numbers $\mathbb R_+$. Note that the notation $P_t(g_p)$ indicates that the ST has the perfect knowledge of the channel gain $g_p$; in the sequel, when we consider the case where the ST has no knowledge or only statistical knowledge of the channel gain $g_p$, we will denote the power transmitted from the ST simply by $P_t$. 

After receiving the signals, the user $U_u, u \in \{n, f\}$ first combines the signals using maximal-ratio combining (MRC), and therefore, the channel gain between the ST and $U_u$ is given by $g_u \triangleq \bs h_u^H \bs h_u$, where $\bs h_u = [h_{u, 1} \ h_{u, 2} \cdots \ h_{u, N_u}]^T \in \mathbb C^{N_u \times 1}$. The near user $U_n$ first decodes $x_f$ by considering the inter-user interference due to the presence of $x_n$ in the received signal as additional noise. It then applies successive interference cancellation (SIC) to remove $x_f$ from the received signal and then decodes its intended symbol $x_n$. On the other hand, the far user $U_f$ decodes $x_f$ directly considering the interference due to $x_n$ as additional noise. Assuming the noise contributions at all receiver nodes to be distributed as $\mathcal{CN}(0, 1)$, the instantaneous signal-to-interference-plus-noise ratio (SINR) and instantaneous signal-to-noise ratio (SNR) at $U_n$ to decode $x_f$ and $x_n$ are, respectively, given by 
\begin{align*}
	\gamma_n^{(f)} = \dfrac{a_f g_n P_t(g_p)}{a_n g_n P_t(g_p) + 1}, \qquad \gamma_n^{(n)} = a_n g_n P_t(g_p).
\end{align*}
Similarly, the instantaneous SINR at $U_f$ to decode $x_f$ is given by 
\begin{align*}
	\gamma_f^{(f)} = \dfrac{a_f g_f P_t(g_p)}{a_n g_f P_t(g_p) + 1}.
\end{align*}
Since symbol $x_n$ needs to be decoded only by $U_n$, the instantaneous achievable rate for $U_n$ is given by 
\begin{align*}
	\log_2\left[1 + \gamma_n^{(n)}\right] = \log_2 [1 + a_n g_p P_t(g_p)].
\end{align*}
On the other hand, since $x_f$ needs to be decoded by both users, the instantaneous achievable rate for $U_f$ is given by 
\begin{align*}
	& \min \left\{ \log_2 \left[ 1 + \gamma_n^{(f)}\right], \log_2 \left[ 1 + \gamma_f^{(f)}\right] \right\} \\
	= & \ \log_2 \left[ 1 + \dfrac{a_f g_{\min} P_t(g_p)}{a_n g_{\min} P_t(g_p) + 1}\right],
\end{align*}
where $g_{\min} \triangleq \min \{g_n, g_f\}$. 

In contrast to this, for the case of OMA, the ST transmits $\sqrt{P_t(g_p)} x_n$ and $\sqrt{P_t(g_p)} x_f$ to $U_n$ and $U_f$, respectively, in two orthogonal time slots. Therefore, the instantaneous SNR at $U_n$ and $U_f$ to decode the intended symbol is, respectively, given by 
\begin{align*}
	\hat \gamma_n = g_n P_t(g_p), \qquad \hat \gamma_f = g_f P_t(g_p).
\end{align*}
Throughout this paper, $f_{\mathcal X}(\cdot)$, $F_{\mathcal X}(\cdot)$, $F^{-1}_{\mathcal X}(\cdot)$ and $\mathcal F_{\mathcal X}(\cdot)$ denote the probability density function (PDF), cumulative distribution function (CDF), inverse distribution function (IDF), and complementary CDF (CCDF) of the random variable $\mathcal X$, respectively.

Next, we will present the achievable rate, outage probability and optimal power allocation for the spectrum sharing system.
\section{Secondary Performance for IntICSI} \label{Sec_IntCSI}
In this section, we assume that the ST has perfect instantaneous IL-CSI, and adapts its transmission power such that the instantaneous interference caused by the ST at the PR is less than a predefined threshold value $I$. In addition, we do not consider any power budget limit for the ST. Such a scenario is relevant when the ST is one with unlimited power, such as a base station. 
\subsection{Average achievable sum-rate} \label{Sec_IntSCI_Cap}
The average achievable sum-rate for the NOMA system is given by 
\begin{align}
	& C_{\mathrm{sum}} = \max_{P_t(g_p) \geq 0} \mathbb E_{g_p, g_n, g_f} \Bigg\{ \log_2 \left[ 1 + a_n g_n P_t(g_p)\right] \notag\\
	& \hspace{2.7cm}\left. + \log_2 \left[ 1 + \dfrac{a_f g_{\min} P_t(g_p)}{a_n g_{\min} P_t(g_p) + 1}\right] \right\}, \label{Sec_IntCSI_Csum_Def1}\\
	& \text{s.t. } g_p P_t(g_p) \leq I \label{Sec_IntCSI_PeakInt}.
\end{align}
The optimal transmit power $P_t^*(g_p)$ that maximizes $C_{\mathrm{sum}}$ in~\eqref{Sec_IntCSI_Csum_Def1} is given by $I/g_p$. Therefore, the expression for the average achievable sum-rate is given by 
\begin{multline} \label{Sec_IntCSI_Csum_Def}
	C_{\mathrm{sum}} = \mathbb E_{X_n} \left\{ \log_2 \left(1 + a_n I X_n\right)\right\}  \\
	   + \mathbb E_{X_{\min}} \left\{ \log_2 \left(1  + I X_{\min} \right)\right\} - \mathbb E_{X_{\min}} \left\{ \log_2 \left(1 + a_n I X_{\min}\right) \right\},
\end{multline}
where $X_n \triangleq g_n/g_p$ and $X_{\min} \triangleq g_{\min}/g_p$.
\begin{theorem} \label{Sec_IntCSI_TheoremCap}
	For the case of IntICSI, the average achievable sum-rate for the NOMA system is given by~\eqref{Sec_IntCSI_CapClosed}, shown on the next page, where $\Omega \triangleq \Omega_n \Omega_f / (\Omega_n + \Omega_f)$ and $G(\cdot)$ denotes Meijer's G-function. 
\end{theorem}
\begin{figure*}
\normalsize
\setcounter{MYtempeqncnt}{\value{equation}}
\begin{multline} \label{Sec_IntCSI_CapClosed}
	 C_{\mathrm{sum}} = \dfrac{1}{\ln 2} \left[ \dfrac{1}{\Gamma(N_n)} \left( \dfrac{\Omega_p}{\Omega_n a_n I} \right)^{N_n}  G_{3, 3}^{3, 2} \left( \dfrac{\Omega_p}{\Omega_n a_n I} \left \vert \begin{smallmatrix} -N_n, & - N_n, & 1 - N_n \\[0.6em] 0, & - N_n, & - N_n\end{smallmatrix} \right. \right) + \sum_{k = 0}^{N_f - 1} \dfrac{\Omega_p^{N_n + k}}{\Gamma(N_n) k! \Omega_n^{N_n} \Omega_f^k I^{N_n + k}} \right. \\
	 \times \left\{ G_{3, 3}^{3, 2} \left( \dfrac{\Omega_p}{\Omega I} \left \vert \begin{smallmatrix} - N_n - k, & - N_n - k, & 1 - N_n - k \\[0.6em] 0, & - N_n - k, & - N_n - k\end{smallmatrix} \right. \right) - \dfrac{1}{a_n^{N_n + k}} G_{3, 3}^{3, 2} \left( \dfrac{\Omega_p}{\Omega a_n I} \left \vert \begin{smallmatrix} - N_n - k, & - N_n - k, & 1 - N_n - k \\[0.6em] 0, & - N_n - k, & - N_n - k\end{smallmatrix} \right. \right)\right\} + \sum_{l = 0}^{N_n - 1} \dfrac{\Omega_p^{N_f + l}}{\Gamma(N_f) l!} \\
	 \left. \times \dfrac{1}{\Omega_f^{N_f} \Omega_n^l I^{N_f + l}}\left\{ G_{3, 3}^{3, 2} \left( \dfrac{\Omega_p}{\Omega I} \left \vert \begin{smallmatrix} - N_f - l, & - N_f - l, & 1 - N_f - l \\[0.6em] 0, & - N_f - l, & - N_f - l\end{smallmatrix} \right. \right) - \dfrac{1}{a_n^{N_f + l}} G_{3, 3}^{3, 2} \left( \dfrac{\Omega_p}{\Omega a_n I} \left \vert \begin{smallmatrix} - N_f - l, & - N_f - l, & 1 - N_f - l \\[0.6em] 0, & - N_f - l, & - N_f - l\end{smallmatrix} \right. \right) \right\} \right].
\end{multline}
\setcounter{equation}{\value{MYtempeqncnt}}
\hrulefill
\end{figure*}
\addtocounter{equation}{1}
\begin{proof}
	See Appendix~\ref{Sec_IntCSI_ProofCap}.
\end{proof}
On the other hand, for the case of OMA, the average achievable sum-rate is given by 
\begin{align}
	\mathscr C_{\mathrm{sum}} = 0.5 \sum_{u \in \{n, f\}} \mathbb E_{X_u} \left\{ \log_2 (1 + I X_u)\right\}, \label{Sec_IntCSI_CapOMA}
\end{align}
where $X_u \triangleq g_u/g_p$. Note that for a fair comparison,~in~\eqref{Sec_IntCSI_CapOMA} the peak tolerable interference value $I$ is the same as considered for the NOMA systems. Also, note that one transmission cycle in NOMA takes a single time slot, whereas for OMA (time-division multiplexed system), one transmission cycle takes two time slots; therefore, a factor of 0.5 is included in~\eqref{Sec_IntCSI_CapOMA} for fair comparison between NOMA and OMA systems. We do not provide a closed-form analysis for the case of OMA, as the focus of this paper is on the NOMA-based system. For the purpose of comparison, we will evaluate the performance of the OMA-based system numerically.

\subsubsection*{Asymptotic analysis}
For the limiting case where $I \to \infty$, it is straightforward to observe that the average achievable rate for $U_f$ converges to $\log_2 (1 + a_f/a_n) = \log_2(1/a_n)$. Therefore, the average achievable sum-rate for NOMA can be approximated by 
	\begin{align}
		& C_{\mathrm{sum}} \approx \mathbb E_{X_n} \{\log_2 (a_n I X_n)\} + \log_2\left( \dfrac{1}{a_n}\right) \notag \\
		= & \ \dfrac{N_n}{\Omega_n^{N_n} \Omega_p} \int_0^\infty x^{N_n - 1} \log_2(a_n I x) \left( \dfrac{x}{\Omega_n} + \dfrac{1}{\Omega_p}\right)^{-(N_n+1)} \mathrm dx \notag \\
		& \hspace{6cm}+ \log_2\left( \dfrac{1}{a_n}\right) \notag \\
		= & \ \log_2(e) \left[ \mathfrak E + \ln \left( \dfrac{a_n \Omega_n I}{\Omega_p} \right) + \mathfrak D(N_n) + \ln \left( \dfrac{1}{a_n}\right) \right], \label{Sec_IntCSI_CapClosed_Asymptotic}
	\end{align}
where $\mathfrak E$ is the Euler-Mascheroni constant and $\mathfrak D(\cdot)$ is the digamma function. This leads to a couple of important observations:
	\begin{itemize}
		\item For large $I$, there is no benefit (in terms of achievable rate) from having multiple antennas at the far user.
		\item For large $I$, when the number of antennas at the near user is increased from (say) $N_n^{(1)}$ to $N_n^{(2)}$, the gain in the near user's achievable rate is quantified by $\log_2 (e) \left[ \mathfrak D \left(N_n^{(2)} \right) - \mathfrak D \left( N_n^{(1)} \right) \right]$. This is also equal to the gain in average achievable sum-rate for the NOMA system.
	\end{itemize}

Next we present the analysis of the outage probability for both NOMA and OMA systems.
\subsection{Outage probability} \label{Sec_IntCSI_Out}
We assume that the target data rates for users $U_n$ and $U_f$ are the same, and are denoted by $r_{\mathrm{target}}$. Therefore, for the case of NOMA, the outage threshold is defined as $\theta \triangleq 2^{r_{\mathrm{target}}} - 1$.
\begin{theorem} \label{Sec_IntCSI_TheoremOut}
	For the case of IntICSI, the outage probability for the NOMA system is given by 
	\begin{align}
		\mathbb P_{\mathrm{out}} = 1 - \prod_{u \in \{n, f\}} \left[1 - \left( \dfrac{\Omega_p \xi_u}{\Omega_u I + \Omega_p \xi_u}\right)^{N_u} \right], \label{Sec_IntCSI_OutClosed}
	\end{align}
where $\xi_n \triangleq \theta \max \left\{ \tfrac{1}{a_f - a_n \theta}, \tfrac{1}{a_n}\right\}$	and $\xi_f \triangleq \tfrac{\theta}{a_f - a_n \theta}$. 
\end{theorem}
\begin{proof}
	See Appendix~\ref{Sec_IntCSI_ProofOut}.
\end{proof}
It is noteworthy that the term $a_f - a_n \theta$ in the denominator of $\xi_f$ indicates that we require $a_f > a_n \theta$, i.e., $a_n < 1/(1 + \theta)$, otherwise both $U_n$ and $U_f$ will fail to decode $x_f$ and the outage probability of the system will always be equal to 1. A similar phenomenon was observed in~\cite{Vaibhav} and~\cite{RelaySelectionDing}.

On the other hand, for the case of OMA, the outage threshold is defined as $\Theta \triangleq 2^{2 r_{\mathrm{target}}} - 1$ and the outage probability is given by 
\begin{align}
	\mathscr P_{\mathrm{out}} = 1 - \prod_{u \in \{n, f\}} \Pr \left( I X_u \geq \Theta \right). \label{Sec_IntCSI_OutOMA_Def}
\end{align}
Here we do not provide a closed-form expression for the outage probability for the case of OMA, but we will rather evaluate via simulation for the purpose of comparison.
\subsubsection*{Asymptotic analysis} For the limiting case where $I \to \infty$, it can be shown using the binomial expansion that 
\begin{align*}
& \mathbb P_{\mathrm{out}} = \sum_{k = 0}^{\infty} (-1)^k \binom{N_f \!+ k \!- \! 1}{k} \Psi_f^{N_f + k} I^{-(N_f + k)} + \sum_{l = 0}^{\infty} (-1)^l \\
&  \times \binom{N_n + l - 1}{l} \Psi_n^{N_n + l} I^{-(N_n + l)} - \sum_{k = 0}^{N_f} \sum_{l = 0}^{N_s} (-1)^{k + l} \\
&  \times \binom{N_f \! + \! k \!-\! 1}{k} \binom{N_n\! + \! l \! - \!1}{l} \Psi_f^{N_f + k} \Psi_n^{N_n + l} I^{-(N_f + N_n + k + l)},
\end{align*}
where $\Psi_u \triangleq \Omega_p \xi_u / \Omega_u$. Using the preceding expression, it is straightforward to conclude that $\mathbb P_{\mathrm{out}}$ decays as $I^{-\min\{N_n, N_f\}}$ for large values of $I$.
\subsection{Optimal power allocation} \label{Sec_IntCSI_OptPower}
In this subsection, we attempt to find a closed-form expression for the optimal $a_n$, denoted by $a_n^*$, that \emph{minimizes} the outage probability of the spectrum sharing NOMA system. By differentiating~\eqref{Sec_IntCSI_OutClosed}, it can be observed that a closed-form expression for $a_n^*$ is not possible in general. However, in the following theorem we show that this is possible in the special case $N_n = N_f = 1$.  
\begin{theorem} \label{Sec_IntCSI_TheoremOpt}
	For the case of IntICSI with $N_n = N_f = 1$, $a_n^*$ is given by 
	\begin{align}
		a_n^* = & \ \dfrac{I \Omega_f + \Omega_p \theta}{I \{(1 + \theta) \Omega_f - \Omega_n\}} \notag \\
		& - \dfrac{\sqrt{(1 + \theta) (I \Omega_f + \Omega_p \theta) \{I \Omega_n + \Omega_p \theta (1 + \theta)\}}}{I (1 + \theta) \{(1 + \theta) \Omega_f - \Omega_n\}}. \label{Sec_IntCSI_asStar}
	\end{align}
\end{theorem}
\begin{proof}
	See Appendix~\ref{Sec_IntCSI_ProofOpt}.
\end{proof}
With simple algebraic manipulation, it can be shown that for the case when $N_n = N_f = 1$, the value of $a_n^*$ decreases with an increase in the value of $I$. For the case when $N_n > 1$ and $N_f > 1$, we find the optimal value of $a_n$ numerically. 
\section{Secondary Performance for IntSCSI} \label{Sec_IntNoCSI}
In a practical system, it is often not possible to obtain instantaneous CSI at the transmitter side. Motivated by this issue, we consider the scenario where the ST has only the statistical CSI regarding the ST-PR link, i.e., only the information regarding $\Omega_p$ and the distribution of $h_p$ is available at the ST (along with the statistical CSI of the ST-$U_u$ links). In this case, the quality-of-service (QoS) at the PR is protected through a statistical constraint which states that the probability that the interference caused by the ST to the PR is above the interference threshold $I$ should be lower than a preset threshold $\delta$. Denoting the power transmitted from ST by $P_t$, we have 
\begin{align}
	& \Pr(g_p P_t > I) \leq \delta \notag \\
	\implies & 1 - F_{g_p} (I/P_t) \leq \delta \implies P_t \leq \dfrac{I}{F_{g_p}^{-1}(1 - \delta)}. \label{Sec_IntNoCSI_PtGeneral}
\end{align}
Given that $g_p$ is an exponentially distributed random variable with mean value given by $\Omega_p$, the IDF of $g_p$ is given by $F_{g_p}^{-1}(x) = -\Omega_p \ln(1 - x)$. Substituting the expression for $F_{g_p}^{-1}(\cdot)$ into~\eqref{Sec_IntNoCSI_PtGeneral}, the optimal transmit power to maximize the average achievable sum-rate is given by 
\begin{align}
P_t^* = -I/(\Omega_p \ln \delta). \label{Sec_IntNoCSI_PtOptimal}
\end{align}

Next, we will provide analytical expressions for the average achievable sum-rate, outage probability and optimal power allocation in spectrum sharing NOMA system for the IntSCSI case.
\subsection{Average achievable sum-rate} \label{Sec_IntNoCSI_Cap}
The expression for the average achievable sum-rate in NOMA is be given by 
\begin{multline} \label{Sec_IntNoCSI_Cap_Def}
	C_{\mathrm{sum}} = \mathbb E_{g_n} \left\{ \log_2(1 + a_n g_n P_t^*)\right\}  \\
	 + \mathbb E_{g_{\min}} \left\{ \log_2 (1 + g_{\min} P_t^*)\right\} - \mathbb E_{g_{\min}} \left\{ \log_2 (1 + a_n g_{\min} P_t^*)\right\}. 
\end{multline}
\begin{theorem} \label{Sec_IntNoCSI_TheoremCap}
	For the case of IntSCSI, the average achievable sum-rate for NOMA is given by~\eqref{Sec_IntNoCSI_CapClosed}, shown on the next page.	
\end{theorem}
\begin{figure*}
\normalsize
\setcounter{MYtempeqncnt}{\value{equation}}
\begin{multline} \label{Sec_IntNoCSI_CapClosed}
	 C_{\mathrm{sum}} = \dfrac{1}{\ln 2} \left[ \dfrac{1}{\Gamma(N_n) \Omega_n^{N_n} (a_n P_t^*)^{N_n}} G_{2, 3}^{3, 1} \left( \dfrac{1}{\Omega_n a_n P_t^*} \left\vert \begin{smallmatrix} -N_n, & 1 - N_n \\[0.6em] 0, & - N_n, & - N_n \end{smallmatrix} \right. \right) + \dfrac{1}{\Gamma(N_n) \Omega_n^{N_n}} \sum_{k = 0}^{N_f - 1} \dfrac{1}{k! \Omega_f^k (P_t^*)^{N_n + k}} \right. \\
	 \times \left\{ G_{2, 3}^{3, 1} \left( \dfrac{1}{\Omega  P_t^*} \left\vert \begin{smallmatrix} -N_n-k, & 1 - N_n -k \\[0.6em] 0, & - N_n -k, & - N_n -k \end{smallmatrix} \right. \right) - \dfrac{1}{a_n^{N_n + k}}  G_{2, 3}^{3, 1} \left( \dfrac{1}{\Omega  a_n P_t^*} \left\vert \begin{smallmatrix} -N_n-k, & 1 - N_n -k \\[0.6em] 0, & - N_n -k, & - N_n -k \end{smallmatrix} \right. \right) \right\}  + \dfrac{1}{\Gamma(N_f) \Omega_f^{N_f}} \\
	 \times \left. \sum_{l = 0}^{N_n - 1} \dfrac{1}{l! \Omega_n^l (P_t^*)^{N_f + l}} \left\{ G_{2, 3}^{3, 1} \left( \dfrac{1}{\Omega  P_t^*} \left\vert \begin{smallmatrix} -N_f-l, & 1 - N_f -l \\[0.6em] 0, & - N_f -l, & - N_f - l \end{smallmatrix} \right. \right) - \dfrac{1}{a_n^{N_f + l}} G_{2, 3}^{3, 1} \left( \dfrac{1}{\Omega  a_n P_t^*} \left\vert \begin{smallmatrix} -N_f-l, & 1 - N_f -l \\[0.6em] 0, & - N_f -l, & - N_f - l \end{smallmatrix} \right. \right) \right\} \right].
\end{multline}
\setcounter{equation}{\value{MYtempeqncnt}}
\hrulefill
\end{figure*}
\addtocounter{equation}{1}
\begin{proof}
	See Appendix~\ref{Sec_IntNoCSI_ProofCap}.
\end{proof}
On the other hand, for the case of OMA, the expression for the average achievable sum-rate is given by 
\begin{align}
	\mathscr C_{\mathrm{sum}} = 0.5 \sum_{u \in \{n, f\}} \mathbb E_{g_u} \left\{ \log_2\left(1 - \dfrac{g_u I}{\Omega_p \ln \delta}\right)\right\}.  \label{Sec_IntNoCSI_CapOMA}
\end{align}
\subsubsection*{Asymptotic analysis}
Similar to the IntICSI case, for large~$I$, the average achievable sum-rate for NOMA can be approximated by 
	\begin{align}
		& C_{\mathrm{sum}} = \mathbb E_{g_n} \left\{ \log_2 \left( a_n g_n \dfrac{I}{\Omega_p (-\ln \delta)} \right) \right\} + \log_2 \left( \dfrac{1}{a_n}\right) \notag \\
		= & \dfrac{1}{\Gamma(N_n) \Omega_n^{N_n}} \int_0^\infty \log_2 \left( \dfrac{a_n I}{\Omega_p (-\ln \delta)} x\right) x^{N_n - 1} \notag \\
		& \hspace{3cm}\times \exp \left( \dfrac{-x}{\Omega_n} \right) \mathrm dx + \log_2 \left( \dfrac{1}{a_n} \right) \notag \\
		= & \ \log_2(e) \left[ \ln \left( \dfrac{a_n \Omega_n I}{\Omega_p (-\ln \delta)} \right) + \mathfrak D (N_n) + \ln \left( \dfrac{1}{a_n} \right)\right]. \label{Sec_IntNoCSI_CapClosed_Asymptotic}
	\end{align}
Similar to the case of IntICSI, for large $I$, there is no benefit (in terms of achievable rate) from having multiple antennas at the far user. Also, for large $I$, when the number of antennas at the near user is increased from $N_n^{(1)}$ to $N_n^{(2)}$, the gain in average achievable sum-rate can be quantified by $\log_2(e) \left[ \mathfrak D \left( N_n^{(2)}\right) - \mathfrak D \left( N_n^{(1)}\right) \right]$, which is the same as for the case of IntICSI.  
\subsection{Outage probability} \label{Sec_IntNoCSI_Out}
Following similar arguments as used in~Section~\ref{Sec_IntCSI_Out}, the outage probability for NOMA is given by 
\begin{align}
	\mathbb P_{\mathrm{out}} = & \ 1 - \Pr \left( \dfrac{a_f P_t^* g_n}{1 + a_n P_t^* g_n} \geq \theta, \ a_n P_t^* g_n \geq \theta \right) \notag \\
	& \hspace{3.2cm}\times \Pr \left( \dfrac{a_f P_t^* g_f}{1 + a_n P_t^* g_f} \geq \theta \right) \notag \\
	= & \ 1 - \Pr \left( g_n \geq \dfrac{\theta}{P_t^*} \max \left\{ \dfrac{1}{a_f - a_n \theta}, \ \dfrac{1}{a_n}\right\} \right) \notag \\
	& \hspace{3.2cm} \times \Pr \left( g_f \geq \dfrac{\theta}{P_t^* (a_f - a_n \theta)}\right) \notag \\
	= & \ 1 - \prod_{u \in \{n, f\}} \mathcal F_{g_u} \left(\dfrac{\xi_u}{P_t^*} \right) \notag \\
	= & \ 1 - \prod_{u \in \{n, f\}} \dfrac{1}{\Gamma(N_u) \Omega_u^{N_u}}\int_{\xi_u/P_t^*}^{\infty} x^{N_u - 1} \exp \left( \dfrac{-x}{\Omega_u} \right) \mathrm dx \notag \\
	= & \ 1 - \prod_{u \in \{n, f\}}  \dfrac{\Gamma[N_u, \xi_u/(\Omega_u P_t^*)]}{\Gamma(N_u)},  \label{Sec_IntNoCSI_OutClosed}
\end{align}
where the integral above is solved using~\cite[eqn.~(3.381-3),~p.~346]{Grad} and $\Gamma[\cdot, \cdot]$ denotes the upper-incomplete Gamma function.

On the other hand, for the case of OMA, the outage probability is given by 
\begin{align}
	\mathscr P_{\mathrm{out}} = & \ 1 - \prod_{u \in \{n, f\}} \Pr \left( P_t^* g_u \geq \Theta \right). \label{Sec_IntNoCSI_OutOMA}
\end{align}

\subsubsection*{Asymptotic analysis} Using~\cite[eqn.~(8.7.2),~p.~178]{NIST}, it can be shown that 
\begin{align*}
& \mathbb P_{\mathrm{out}} = \dfrac{1}{\Gamma(N_f)} \sum_{k = 0}^{\infty} \dfrac{(-1)^k \Upsilon_f^{N_f + k}}{k! (N_f + k)} I^{-(N_f + k)} + \dfrac{1}{\Gamma(N_n)}  \\
& \times \sum_{l = 0}^{\infty} \dfrac{(-1)^l \Upsilon_n^{N_n + l}}{l! (N_n + l)} I^{-(N_n + l)} - \dfrac{1}{\Gamma(N_f) \Gamma(N_n)}  \\
& \times \sum_{k = 0}^{\infty} \sum_{l = 0}^{\infty}\dfrac{(-1)^{k + l} \Upsilon_f^{N_f + k} \Upsilon_n^{N_n + l}}{k! l! (N_f + k) (N_n + l)} I^{-(N_f + N_n + k + l)}, 
\end{align*}
where $\Upsilon_u \triangleq -\xi_u \Omega_p \ln \delta /\Omega_u$. From the preceding equation, it is straightforward to conclude that $\mathbb P_{\mathrm{out}}$ decays as $I^{-\min\{N_n, N_f\}}$ for large values of $I$.

\subsection{Optimal power allocation} \label{Sec_IntNoCSI_OptPower}
\begin{theorem} \label{Sec_IntNoCSI_TheoremOpt}
For the case of IntSCSI with $N_n = N_f = 1$, $a_n^*$ is given by 
\begin{align}
		a_n^* = & \dfrac{\Omega_f}{(1 + \theta) \Omega_f - \Omega_n} - \dfrac{\sqrt{\Omega_n \Omega_f (1 + \theta)}}{(1 + \theta) \left\{(1 + \theta) \Omega_f -\Omega_n \right\}}. \label{Sec_IntNoCSI_asStar}
\end{align}
\end{theorem}
\begin{proof}
	See~Appendix~\ref{Sec_IntNoCSI_ProofOpt}.
\end{proof}
\noindent It is important to note that in this case, the optimal value of $a_n$ does not depends on $I$ or $\Omega_p$.

Next, we present the analysis for the spectrum sharing system where a power budget constraint also exists at the ST, along with a peak interference constraint at the PR.
\section{Secondary Performance for PowIntICSI} \label{Sec_PowIntCSI}
For the case when the ST is a battery-operated device, the power transmitted from the ST is often constrained by a peak power budget at the ST. Therefore, in this section, we analyze the performance of the spectrum sharing system where the power transmitted from the ST is constrained by the peak interference caused at the PR as well as a peak power budget at the ST.
\subsection{Average achievable sum-rate}\label{Sec_PowIntCSI_Cap}
The average achievable sum-rate for the NOMA system is given by 
\begin{align}
	& C_{\mathrm{sum}} = \max_{P_t(g_p) \geq 0} \mathbb E_{g_p, g_n, g_f} \Bigg\{ \log_2[1 + a_n g_n P_t(g_p)] \notag \\
	& \hspace{2.5cm}\left. + \log_2 \left[ 1 + \dfrac{a_f g_{\min} P_t(g_p)}{a_n g_{\min} P_t(g_p) + 1} \right] \right\}, \label{Sec_PowIntCSI_Csum_Def1}\\
	& \text{s.t. } g_p P_t(g_p) \leq I  \label{Sec_PowIntCSI_PeakInt}, \\
	& \hspace{0.55cm}P_t(g_p) \leq P_{\mathrm{peak}}, \label{Sec_PowIntCSI_PeakPow}
\end{align}
where $P_{\mathrm{peak}}$ denotes the peak power budget at the ST.  Therefore, the optimal power to maximize the average achievable sum-rate in NOMA is given by 
\begin{align}
	P_t^*(g_p) = & \ \min \left\{ P_{\mathrm{peak}}, \dfrac{I}{g_p}\right\} = \begin{cases} P_{\mathrm{peak}}, & \text{if } g_p \leq \dfrac{I}{P_{\mathrm{peak}}} \\
	\dfrac{I}{g_p}, & \text{otherwise.} \end{cases} \label{Sec_PowIntCSI_PtOptimal}
\end{align}
Therefore, using~\eqref{Sec_PowIntCSI_Csum_Def1}-\eqref{Sec_PowIntCSI_PtOptimal}, the expression for the average achievable sum-rate for NOMA is given by 
\begin{align}
	C_{\mathrm{sum}} = & \ \mathbb E_{g_p, g_n} \left\{ \log_2 [1 + a_n g_n P_t^*(g_p)]\right\} \notag \\
	& + \mathbb E_{g_p, g_{\min}} \left\{ \log_2[1 + P_t^* (g_p)g_{\min}] \right\} \notag \\
	& - \mathbb E_{g_p, g_{\min}} \left\{ \log_2 [1 + a_n P_t^*(g_p)g_{\min}]\right\}, \label{Sec_PowIntCSI_Cap_Def2}
\end{align}
where $P_t^*(g_p)$ is given by~\eqref{Sec_PowIntCSI_PtOptimal}. It can be shown that in general, it is very difficult (if not impossible) to find an analytical expression for~\eqref{Sec_PowIntCSI_Cap_Def2}. Therefore, we present the analytical expression for the special case where $N_n = N_f = 1$. 
\begin{theorem} \label{Sec_PowIntCSI_TheoremCap}
	For the case of PowIntICSI with $N_n = N_f = 1$, the average achievable sum-rate for NOMA is given by 
	\begin{align}
		C_{\mathrm{sum}} = \dfrac{1}{\ln 2} \left[ T(a_n \Omega_n) + T(\Omega) - T(a_n \Omega)\right], \label{Sec_PowIntCSI_CapClosed}
	\end{align}
where $T(x)$ is given in~\eqref{Sec_PowIntCSI_T} (shown on the next page) in which $\operatorname{Shi}(\cdot)$ and $\operatorname{Chi}(\cdot)$ denote the hyperbolic sine and hyperbolic cosine integrals, respectively.
\end{theorem}
\begin{figure*}[t]
\normalsize
\setcounter{MYtempeqncnt}{\value{equation}}
\begin{multline} \label{Sec_PowIntCSI_T}
	T(x) = - \left[ 1 - \exp \left( \dfrac{-I}{\Omega_p P_{\mathrm{peak}}}\right)\right] \exp \left( \dfrac{1}{x P_{\mathrm{peak}}}\right) \operatorname{Ei} \left( \dfrac{-1}{x P_{\mathrm{peak}}} \right) + \left[ \operatorname{Ei} \left( \dfrac{-I}{\Omega_p P_{\mathrm{peak}}}\right) -\dfrac{\Omega_p}{\Omega_p - x I} \left\{\operatorname{Ei} \left( \dfrac{-I}{\Omega_p P_{\mathrm{peak}}} \right) \right. \right. \\
  \left. \left.  - \exp \left( \dfrac{\Omega_p - x I}{x \ \Omega_p P_{\mathrm{peak}}}\right) \operatorname{Ei}\left( \dfrac{-1}{x P_{\mathrm{peak}}}\right)\right\}  + \exp \left( \dfrac{\Omega_p - x I}{x \ \Omega_p P_{\mathrm{peak}}}\right) \left\{\operatorname{Shi}\left( \dfrac{1}{x P_{\mathrm{peak}}}\right) - \operatorname{Chi}\left( \dfrac{1}{x P_{\mathrm{peak}}}\right) 	\right\} \right].
\end{multline}
\setcounter{equation}{\value{MYtempeqncnt}}
\hrulefill
\end{figure*}
\addtocounter{equation}{1}
\begin{proof}
	See Appendix~\ref{Sec_PowIntCSI_ProofCap}.
\end{proof}

On the other hand, the corresponding average achievable sum-rate for OMA is given by 
\begin{align}
	\mathscr C_{\mathrm{sum}} = 0.5 \sum_{u \in \{n, f\}} \mathbb E_{g_u} \left\{ \log_2\left( 1 + g_u P_t^*(g_p) \right)\right\}. \label{Sec_PowIntCSI_CapOMA}
\end{align}
\subsubsection*{Asymptotic analysis} 
For the limiting case when $I \to \infty$, we give an analytical expression for the more general case where $N_n \geq 1$ and $N_f \geq 1$. In this case, for a finite value of $P_{\mathrm{peak}}$, we have 
\begin{align*}
\lim_{I \to \infty} P_t^*(g_p) = \lim_{I \to \infty} \min \left\{P_{\mathrm{peak}}, \frac{I}{g_p} \right\} = P_{\mathrm{peak}}.
\end{align*}
Therefore, the asymptotic expression for the average achievable sum-rate is given by 
	\begin{align*}
		& C_{\mathrm{sum}} = \mathbb E_{g_n} \left\{ \log_2 [1 + a_n g_n P_{\mathrm{peak}} ]\right\} \\
		& + \mathbb E_{g_{\min}} \!\left\{ \log_2[1 \!+\! g_{\min} P_{\mathrm{peak}}] \right\} \!-\! \mathbb E_{g_{\min}} \!\left\{ \log_2 [1\! +\! a_n g_{\min} P_{\mathrm{peak}}]\right\}.
	\end{align*}
It can be noted that the preceding equation is the same as~\eqref{Sec_IntNoCSI_Cap_Def}, with $P_t^*$ replaced by $P_{\mathrm{peak}}$. Therefore, a closed-form expression for the preceding equation can be given by~\eqref{Sec_IntNoCSI_CapClosed}, with $P_t^*$ replaced by $P_{\mathrm{peak}}$. 
\subsection{Outage probability}
Following the arguments in~Section~\ref{Sec_IntCSI_Out}, the outage probability for the case of NOMA is defined as 
\begin{align}
	\mathbb P_{\mathrm{out}} = & 1 - \Pr \left( g_n \geq \dfrac{\xi_n}{P_t^*(g_p)}\right) \Pr \left( g_f \geq \dfrac{\xi_f}{P_t^*(g_p)}\right). \label{Sec_PowIntCSI_Out_Def}
\end{align}
\begin{theorem} \label{Sec_PowIntCSI_TheoremOut}
	For the case of PowIntICSI, the outage probability for NOMA is given by~\eqref{Sec_PowIntCSI_OutClosed}, shown on the next page.
\end{theorem}
\begin{figure*}[t]
\normalsize
\setcounter{MYtempeqncnt}{\value{equation}}
\begin{multline} \label{Sec_PowIntCSI_OutClosed}
	\mathbb P_{\mathrm{out}} = 1 - \left\{ \left[ 1 - \exp \left( \dfrac{-I}{\Omega_p P_{\mathrm{peak}}}\right)\right] \prod_{u \in \{n, f\}} \dfrac{\Gamma[N_u, \xi_u/(\Omega_u P_{\mathrm{peak}})]}{\Gamma(N_u)} \right. \\
	\left.+ \dfrac{1}{\Omega_p} \sum_{k = 0}^{N_f - 1} \sum_{l = 0}^{N_n - 1} \dfrac{\xi_n^l \xi_f^k}{k! l! \Omega_n^l \Omega_f^k I^{k + l}}  \Gamma \left[ k + l + 1, \left( \dfrac{1}{\Omega_p} + \dfrac{\xi_n}{\Omega_n I} + \dfrac{\xi_f}{\Omega_f I} \right) \dfrac{I}{P_{\mathrm{peak}}}\right] \left( \dfrac{1}{\Omega_p} + \dfrac{\xi_n}{\Omega_n I} + \dfrac{\xi_f}{\Omega_f I} \right)^{-(k + l + 1)} \right\}.
\end{multline}
\setcounter{equation}{\value{MYtempeqncnt}}
\hrulefill
\end{figure*}
\addtocounter{equation}{1}
\begin{proof}
	See Appendix~\ref{Sec_PowIntCSI_ProofOut}.
\end{proof}

On the other hand, for the case of OMA, the outage probability is given by
\begin{align}
	\mathscr P_{\mathrm{out}} = 1 - \prod_{u \in \{n, f\}}\Pr \left( g_u \geq \dfrac{\Theta}{P_t^*(g_p)}\right). \label{Sec_PowIntCSI_OutOMA}
\end{align}
\subsubsection*{Asymptotic analysis} 
For the limiting case when $I \to \infty$ and $P_{\mathrm{peak}}$ is finite, $P_t^*(g_p)$ in~\eqref{Sec_PowIntCSI_PtOptimal} becomes equal to $P_{\mathrm{peak}}$. Using~\eqref{Sec_PowIntCSI_Out_Def},~\eqref{Sec_PowIntCSI_Out_Def2} and~\eqref{mathfrakX1}, it is straightforward to conclude that 
\begin{align}
	\mathbb P_{\mathrm{out}} = 1 - \prod_{u \in \{n, f\}} \dfrac{\Gamma[N_n, \xi_u/(\Omega_u P_{\mathrm{peak}})]}{\Gamma(N_u)}, \label{Sec_PowIntCSI_OutClosed_Asymptotic}
\end{align}
which is independent of the peak tolerable interference power~$I$. Therefore, the slope of the outage probability curve tends to zero for large values of~$I$, resulting in a outage floor. The asymptotic expression shown above also helps to quantify the reduction in the outage floor resulting from an increase in the number of antennas at the SRs.
\subsection{Optimal power allocation} \label{Sec_PowIntCSI_OptPower}
It can be shown that for the case of PowIntICSI, it is very complicated (if not impossible) to find a closed-form expression for $a_n^*$, even for the case where $N_n = N_f = 1$. Therefore, for the case of PowIntICSI, we find the optimal value of $a_n$ numerically. 
\section{Secondary Performance for PowIntSCSI} \label{Sec_PowIntNoCSI}
For the analysis in this section, we assume that the power transmitted from the ST is constrained by the peak interference constraint at the PR as well as the peak power budget at the ST. Additionally, we assume that only statistical IL-CSI is available at the ST.
\subsection{Average achievable sum-rate} \label{Sec_PowIntNoCSI_Cap}
The average achievable sum-rate for the case of NOMA is given by
\begin{align}
	& C_{\mathrm{sum}} = \max_{P_t \geq 0} \mathbb E_{g_n, g_f} \left\{ \log_2 (1 + a_n P_t g_n) + \log_2 (1 + P_t g_{\min}) \right. \notag \\
	& \hspace{4cm}\left. - \log_2 (1 + a_n P_t g_{\min})\right\}, \label{Sec_PowIntNoCSI_CapDef_Basic}\\
	& \text{s.t. } \Pr(g_p P_t \geq I) \leq \delta, \label{Sec_PowIntNoCSI_IntConst}\\
	& \hspace{0.6cm} P_t \leq P_{\mathrm{peak}}. \label{Sec_PowIntNoCSI_PeakPow}
\end{align}
Using~\eqref{Sec_PowIntNoCSI_IntConst},~\eqref{Sec_PowIntNoCSI_PeakPow} and~\eqref{Sec_IntNoCSI_PtOptimal}, the optimal transmit power to maximize the sum-rate for NOMA is given by 
\begin{align}
	P_t^* = \min \left\{ P_{\mathrm{peak}}, \dfrac{-I}{\Omega_p \ln \delta}\right\}. \label{Sec_PowIntNoCSI_PtOptimal}
\end{align}
Therefore, using~\eqref{Sec_PowIntNoCSI_CapDef_Basic}-\eqref{Sec_PowIntNoCSI_PtOptimal}, an expression for the average achievable sum-rate for NOMA is given by
\begin{multline} \label{Sec_PowIntNoCSI_Cap_Def}
	C_{\mathrm{sum}} = \mathbb E_{g_n} \left\{ \log_2(1 + a_n g_n P_t^*)\right\}  \\
	 + \mathbb E_{g_{\min}} \left\{ \log_2 (1 + g_{\min} P_t^*)\right\} - \mathbb E_{g_{\min}} \left\{ \log_2 (1 + a_n g_{\min} P_t^*)\right\}. 
\end{multline}
Note that~\eqref{Sec_PowIntNoCSI_Cap_Def} is same as~\eqref{Sec_IntNoCSI_Cap_Def}, however, the definition of $P_t^*$ in~\eqref{Sec_PowIntNoCSI_Cap_Def} and~\eqref{Sec_IntNoCSI_Cap_Def} are different. Therefore, an analytical expression for~\eqref{Sec_PowIntNoCSI_Cap_Def} is given by~\eqref{Sec_IntNoCSI_CapClosed}, with $P_t^*$ given by~\eqref{Sec_PowIntNoCSI_PtOptimal}.

On the other hand, for the case of OMA, the average achievable sum-rate is given by 
\begin{align}
	\mathscr C_{\mathrm{sum}} = 0.5 \sum_{u \in \{n, f\}} \mathbb E_{g_u} \left\{ \log_2 (1 + g_u P_t^*)\right\}, \label{Sec_PowIntNoCSI_CapOMA}
\end{align}
where $P_t^*$ is given by~\eqref{Sec_PowIntNoCSI_PtOptimal}.
\subsubsection*{Asymptotic analysis}
Similar to the case of PowIntICSI, a closed-form expression for the asymptotic achievable sum-rate for the PowIntSCSI NOMA system can be given by~\eqref{Sec_IntNoCSI_CapClosed}, with $P_t^*$ replaced by $P_{\mathrm{peak}}$. This means that for the case of NOMA, the average achievable sum-rate for both PowIntICSI and PowIntSCSI systems are exactly the same for a sufficiently large value of $I$ and a finite value of $P_{\mathrm{peak}}$. This is intuitive because when the value of $I$ is sufficiently large, the (average) transmitted power from the ST for both systems is the same (equal to $P_{\mathrm{peak}}$) and no gain is obtained from the instantaneous knowledge of the ST-PR link CSI.
\subsection{Outage probability} \label{Sec_PowIntNoCSI_Out}
Following arguments similar to those in the previous subsection, the outage probability for NOMA is given by~\eqref{Sec_IntNoCSI_OutClosed}, where $P_t^*$ is given by~\eqref{Sec_PowIntNoCSI_PtOptimal}.

On the other hand, for the case of OMA, the outage probability is given by 
\begin{align}
	\mathscr P_{\mathrm{out}} = 1 - \prod_{u \in \{n, f\}} \Pr \left( g_u \geq \dfrac{\Theta}{P_t^*}\right). \label{Sec_PowIntNoCSI_OutOMA}
\end{align}
\subsubsection*{Asymptotic analysis} Following~\eqref{Sec_IntNoCSI_OutClosed} and~\eqref{Sec_PowIntNoCSI_PtOptimal}, an asymptotic expression for the outage probability for the case when~$I \to \infty$ is given by 
	\begin{align}
	\mathbb P_{\mathrm{out}} = 1 - \prod_{u \in \{n, f\}} \dfrac{\Gamma[N_n, \xi_u/(\Omega_u P_{\mathrm{peak}})]}{\Gamma(N_u)}, \label{Sec_PowIntNoCSI_OutClosed_Asymptotic}
	\end{align}
which is the same as for the case of PowIntICSI. This means that the for the case of NOMA, the outage probability for the PowIntICSI and PowIntSCSI systems are exactly the same for a sufficiently large value of $I$ and a finite value of $P_{\mathrm{peak}}$. 
\subsection{Optimal power allocation} \label{Sec_PowIntNoCSI_OptPower}
Since the expression for the outage probability for the NOMA system in the case of PowIntSCSI is the same as that for the case of IntSCSI with the expression for $P_t^*$ given by~\eqref{Sec_PowIntCSI_PtOptimal} which is independent of $a_n$, therefore, following the arguments in~Section~\ref{Sec_IntNoCSI_OptPower}, for the case where $N_n = N_f = 1$, $a_n^*$ is given by 
\begin{align}
		a_n^* = & \dfrac{\Omega_f}{(1 + \theta) \Omega_f - \Omega_n} - \dfrac{\sqrt{\Omega_n \Omega_f (1 + \theta)}}{(1 + \theta) \left\{(1 + \theta) \Omega_f -\Omega_n \right\}}. \label{Sec_PowIntNoCSI_asStar}
\end{align}
Similar to case in~Section~\ref{Sec_IntNoCSI_OptPower}, $a_n^*$ does not depends on $P_{\mathrm{peak}}$, $I$ or $\Omega_p$.
\section{Secondary Performance with One-Bit Feedback} \label{Sec_OneBit}
In this section, we consider the scenario where the ST does not have any CSI regarding the ST-PR link. We rather assume that the PR has instantaneous CSI regarding the ST-PR link. Also, we assume that the power transmitted from the ST is constrained by a peak interference constraint at the PR, as well as a peak power budget constraint at the ST. 

Based on the peak power budget at the ST, and the peak interference constraint at the PR, the PR calculates a threshold value $\tau$ for the channel gain $g_p$. If the instantaneous channel gain of the ST-PR link is less than $\tau$, the PR sends a ``$1$'' to the ST via a low-bandwidth zero-delay feedback link, and sends a ``$0$'' otherwise. For the case when ST receives a ``1'' from the PR, it transmits its signal to $U_n$ and $U_f$ with full power $P_{\mathrm{peak}}$, otherwise it remains silent. Therefore, the transmit power from the ST is modeled as 
\begin{align}
	P_t^* = \begin{cases} P_{\mathrm{peak}}, & \text{if } g_p \leq \tau = \dfrac{I}{P_{\mathrm{peak}}} \\ 0, & \text{otherwise}.\end{cases} \label{Sec_OneBit_OptimalPower}
\end{align}
Note that the power transmission scheme in~\eqref{Sec_OneBit_OptimalPower} ensures that the interference caused by the ST at the PR is either less than or equal to the peak tolerable interference at the PR or zero.
\subsection{Average achievable sum-rate}
The average achievable sum-rate for the case of NOMA is given by 
\begin{align}
	C_{\mathrm{sum}} = & \mathbb E_{g_n, g_f} \left\{ \log_2(1 + a_n P_t^* g_n) + \log_2(1 + P_t^* g_{\min}) \right. \notag \\
	& \hspace{3cm}\left. - \log_2(1 + a_n P_t^* g_{\min})\right\} \notag \\
	= & \ \Pr (g_p \leq \tau) \left[ \mathbb E_{g_n} \left\{ \log_2(1 + a_n g_n P_{\mathrm{peak}}) \right\} \right. \notag \\
	& \hspace{2cm}+ \mathbb E_{g_{\min}} \left\{ \log_2(1 + g_{\min} P_{\mathrm{peak}}) \right\}  \notag \\
	& \hspace{2cm}\left. - \mathbb E_{g_{\min}} \left\{ \log_2(1 + a_n g_{\min} P_{\mathrm{peak}}) \right\} \right]. \label{Sec_OneBit_CapDef}
\end{align}
Following the steps in~Appendix~\ref{Sec_IntNoCSI_ProofCap}, an analytical expression for~\eqref{Sec_OneBit_CapDef} is given by~\eqref{Sec_OneBit_CapClosed}, shown on the next page. 
\begin{figure*}
\normalsize
\setcounter{MYtempeqncnt}{\value{equation}}
\begin{multline} \label{Sec_OneBit_CapClosed}
	 C_{\mathrm{sum}} = \dfrac{1 - \exp (-\tau/\Omega_p)}{\ln 2} \left[ \dfrac{1}{\Gamma(N_n) \Omega_n^{N_n} (a_n P_{\mathrm{peak}})^{N_n}} G_{2, 3}^{3, 1} \left( \dfrac{1}{\Omega_n a_n P_{\mathrm{peak}}} \left\vert \begin{smallmatrix} -N_n, 1 - N_n \\[0.6em] 0, - N_n, - N_n \end{smallmatrix} \right. \right) + \dfrac{1}{\Gamma(N_n) \Omega_n^{N_n}} \sum_{k = 0}^{N_f - 1} \dfrac{1}{k! \Omega_f^k} \right. \\
	 \times \dfrac{1}{(P_{\mathrm{peak}})^{N_n + k}}\left\{ G_{2, 3}^{3, 1} \left( \dfrac{1}{\Omega  P_{\mathrm{peak}}} \left\vert \begin{smallmatrix} -N_n-k, 1 - N_n -k \\[0.6em] 0, - N_n -k, - N_n -k \end{smallmatrix} \right. \right) - \dfrac{1}{a_n^{N_n + k}}  G_{2, 3}^{3, 1} \left( \dfrac{1}{\Omega  a_n P_{\mathrm{peak}}} \left\vert \begin{smallmatrix} -N_n-k, 1 - N_n -k \\[0.6em] 0,  - N_n -k,  - N_n -k \end{smallmatrix} \right. \right) \right\}  + \dfrac{1}{\Gamma(N_f) \Omega_f^{N_f}} \\
	 \times \left. \sum_{l = 0}^{N_n - 1} \dfrac{1}{l! \Omega_n^l (P_{\mathrm{peak}})^{N_f + l}} \left\{ G_{2, 3}^{3, 1} \left( \dfrac{1}{\Omega  P_{\mathrm{peak}}} \left\vert \begin{smallmatrix} -N_f-l, 1 - N_f -l \\[0.6em] 0, - N_f -l, - N_f - l \end{smallmatrix} \right. \right) - \dfrac{1}{a_n^{N_f + l}} G_{2, 3}^{3, 1} \left( \dfrac{1}{\Omega  a_n P_{\mathrm{peak}}} \left\vert \begin{smallmatrix} -N_f-l, 1 - N_f -l \\[0.6em] 0, - N_f -l, - N_f - l \end{smallmatrix} \right. \right) \right\} \right].
\end{multline}
\setcounter{equation}{\value{MYtempeqncnt}}
\hrulefill
\end{figure*}
\addtocounter{equation}{1}

On the other hand, for the case of OMA, the average achievable sum-rate is given by 
\begin{align}
	\mathscr C_{\mathrm{sum}} = 0.5 \sum_{u \in \{n, f\}} \mathbb E_{g_u} \left\{ \log_2(1 + g_u P_t^*)\right\}, \label{Sec_OneBit_CapOMA}
\end{align}
where $P_t^*$ is given by~\eqref{Sec_OneBit_OptimalPower}.
\begin{figure}[t]
\centering
  \includegraphics[width = 0.9\linewidth]{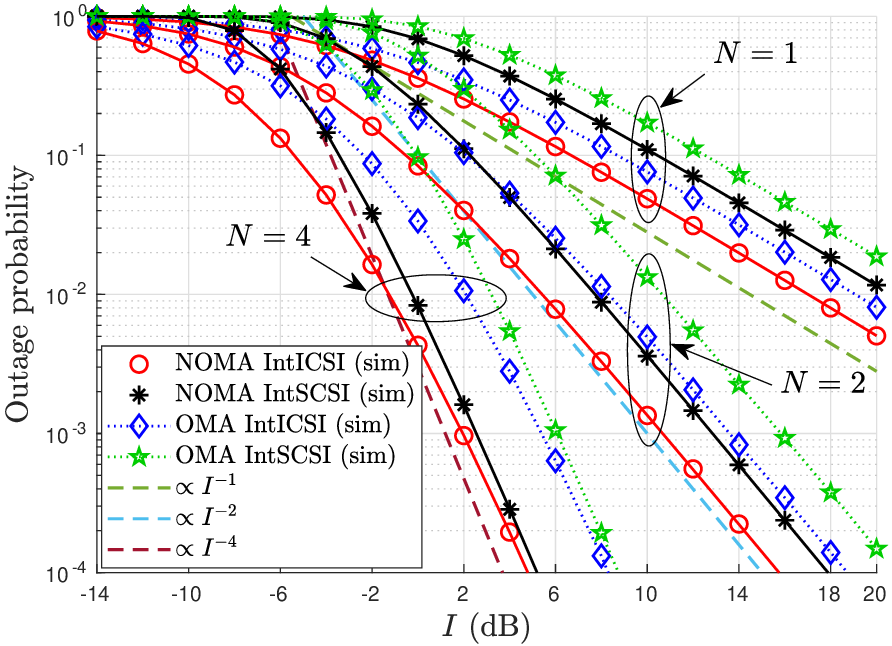}
\caption{Outage probability for the IntICSI and IntSCSI systems. Here solid lines are plotted using~\eqref{Sec_IntCSI_OutClosed} and~\eqref{Sec_IntNoCSI_OutClosed}. The constant-slope dotted lines show the slope of $\mathbb P_{\mathrm{out}}$ for large $I$.}
\label{Fig_OutInt}
\end{figure}

\begin{figure*}[t]
\centering
\begin{minipage}{.32\textwidth}
  \centering
    \includegraphics[width = 0.9\linewidth, height = 4cm]{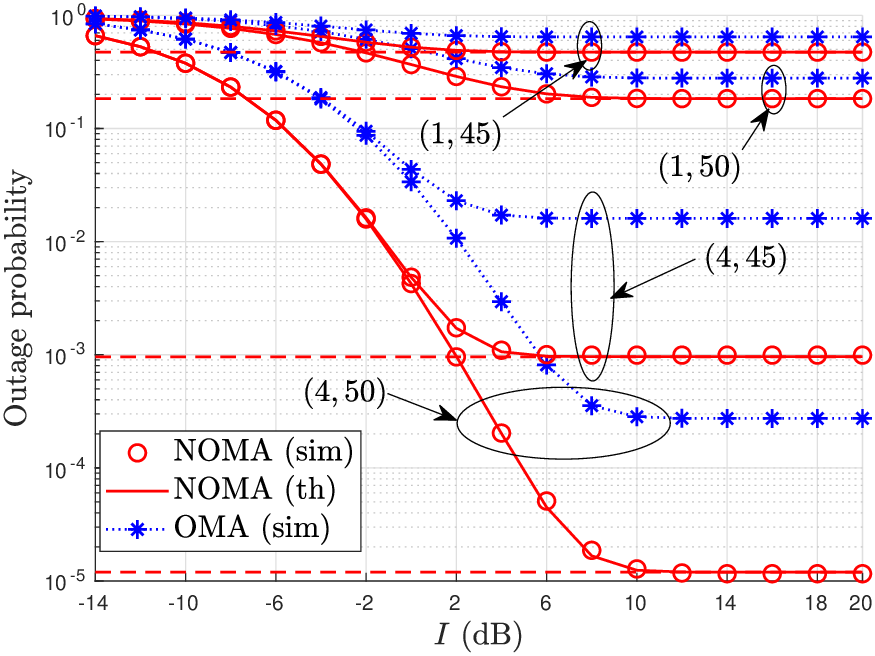}
\caption{Outage probability for the PowIntICSI system. The solid lines are drawn using~\eqref{Sec_PowIntCSI_OutClosed}. The tuples in the parenthesis denote $(N, P_{\mathrm{peak}})$, and the horizontal dotted lines drawn using~\eqref{Sec_PowIntCSI_OutClosed_Asymptotic} denote the asymptotic $\mathbb P_{\mathrm{out}}$ for the NOMA system.}
\label{Fig_OutPowIntCSI}
\end{minipage}
\hfill 
\begin{minipage}{.32\textwidth}
  \centering
	\includegraphics[width = 0.9\linewidth, height = 4cm]{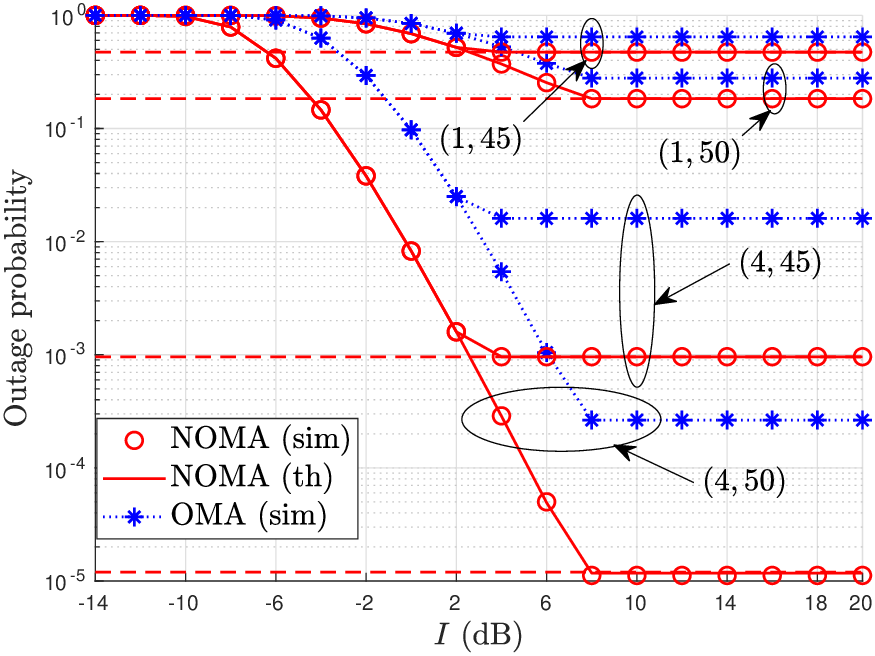}
\caption{Outage probability for the PowIntSCSI system. The solid lines are drawn using~\eqref{Sec_IntNoCSI_OutClosed} with $P_t^*$ given by~\eqref{Sec_PowIntNoCSI_PtOptimal}. The tuples in the parenthesis denote $(N, P_{\mathrm{peak}})$, and the horizontal dotted lines drawn using~\eqref{Sec_PowIntNoCSI_OutClosed_Asymptotic} denote the asymptotic $\mathbb P_{\mathrm{out}}$ for the NOMA system.}
	\label{Fig_OutPowIntNoCSI}
\end{minipage}
\hfill
\begin{minipage}{.32\textwidth}
  \centering
  \includegraphics[width = 0.9\linewidth, height = 4cm]{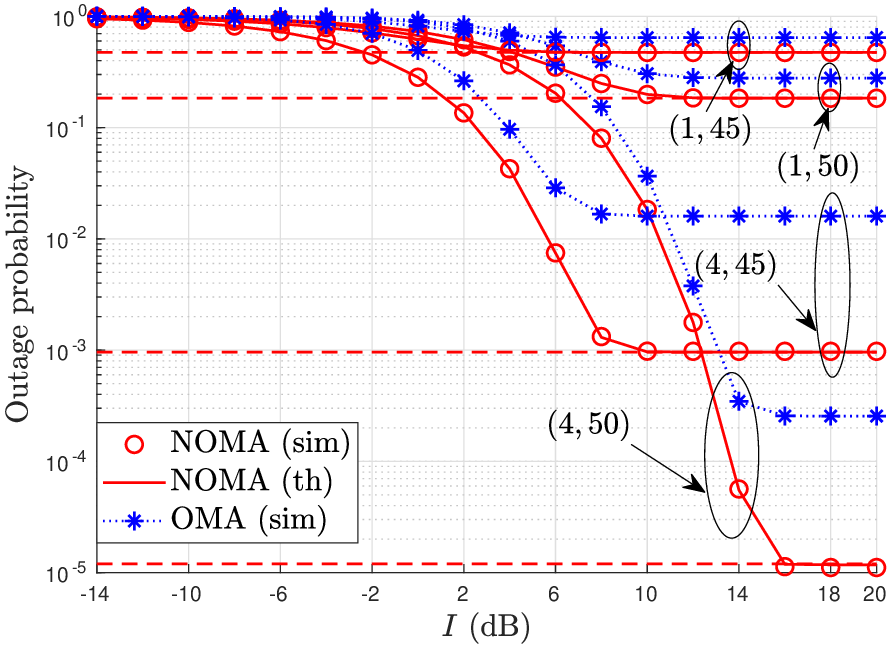}
  \caption{Outage probability for the PowIntOneBit system. The solid lines are drawn using~\eqref{Sec_OneBit_OutClosed}. The tuples in the parenthesis denote $(N, P_{\mathrm{peak}})$, and the horizontal dotted lines drawn using~\eqref{Sec_OneBit_OutClosed_Asymptotic} denote the asymptotic $\mathbb P_{\mathrm{out}}$ for the NOMA system.}
  \label{Fig_OutOneBit}
\end{minipage}%
\end{figure*}

\begin{figure*}[t]
\centering
\begin{minipage}{.32\textwidth}
  \centering
  \includegraphics[width = 0.9\linewidth, height = 4cm]{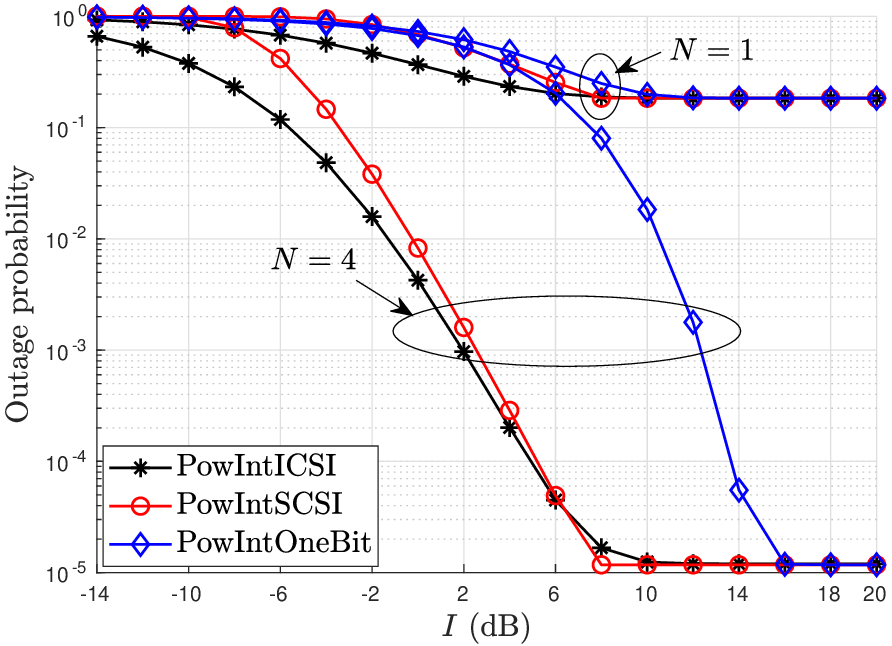}
  \caption{Analytical plots for the outage probability of the PowIntICSI, PowIntSCSI and PowIntOneBit NOMA systems. Here $P_{\mathrm{peak}}$ is fixed at 50 dB.}
  \label{Fig_OutPow}  
\end{minipage}
\hfill 
\begin{minipage}{.32\textwidth}
  \centering
  \includegraphics[width = 0.9\linewidth]{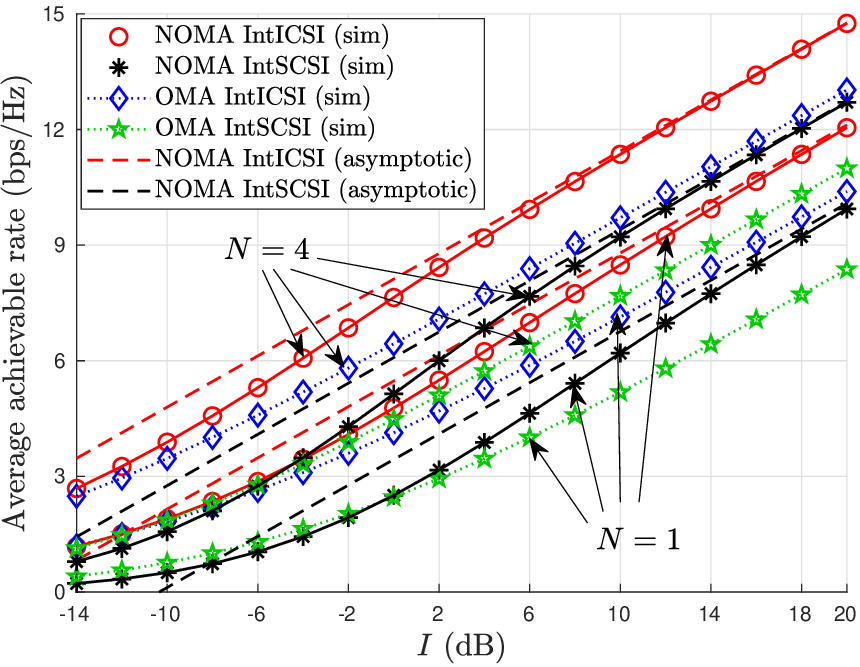}
  \caption{Average achievable sum-rate for the IntICSI and IntSCSI systems. Here solid lines are plotted using~\eqref{Sec_IntCSI_CapClosed} and~\eqref{Sec_IntNoCSI_CapClosed}, while the dotted lines are drawn using~\eqref{Sec_IntCSI_CapClosed_Asymptotic} and~\eqref{Sec_IntNoCSI_CapClosed_Asymptotic}.}
  \label{Fig_CapInt}		
\end{minipage}
\hfill
\begin{minipage}{.32\textwidth}
  \centering
  \includegraphics[width = 0.9\linewidth, height = 4cm]{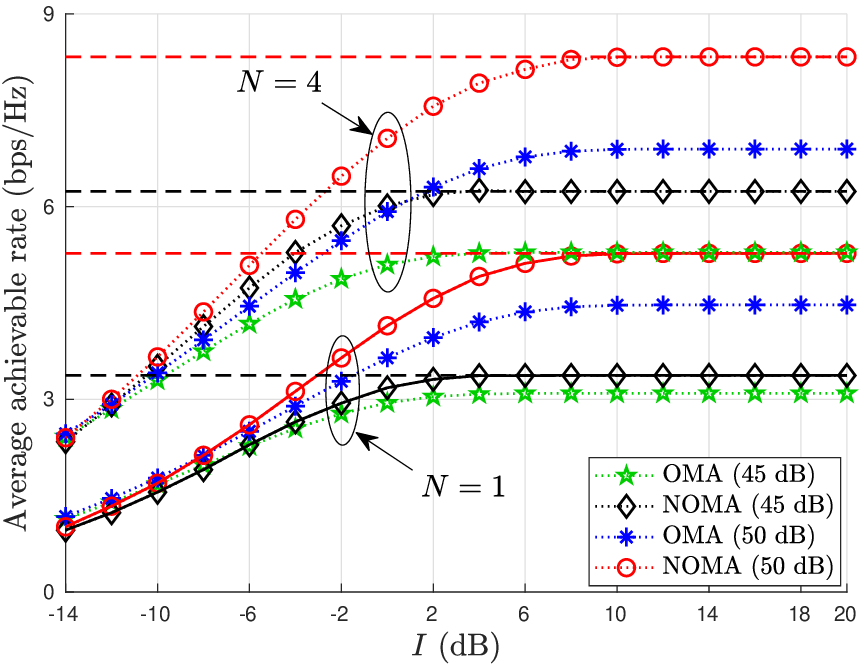}
  \caption{Average achievable sum-rate for the PowIntICSI system. Here markers with dotted lines represent the simulation results, the solid lines drawn using~\eqref{Sec_PowIntCSI_CapClosed} represent the analytical results, and the horizontal dotted lines drawn using~\eqref{Sec_IntNoCSI_CapClosed} with $P_t^*$ replaced by $P_{\mathrm{peak}}$, denote the asymptotic results (for NOMA).}
  \label{Fig_CapPowIntCSI}
\end{minipage}%
\end{figure*}

\subsubsection*{Asymptotic analysis} 
For the case when $I \to \infty$, the transmit power $P_t^*$ in~\eqref{Sec_OneBit_OptimalPower} becomes equal to $P_{\mathrm{peak}}$, and $\Pr(g_p \leq \tau) \to 1$. Therefore, using~\eqref{Sec_OneBit_CapDef}, a closed-form expression for the asymptotic average achievable sum-rate is given by~\eqref{Sec_IntNoCSI_CapClosed}, with $P_t^*$ in~\eqref{Sec_IntNoCSI_CapClosed} replaced by $P_{\mathrm{peak}}$.

This means that for the case when the ST has a fixed power budget, the achievable sum-rate will always converge to the same value for NOMA systems in the high-$I$ regime, regardless of the level of IL-CSI knowledge at the ST.  
\subsection{Outage probability}
Similar to Section~\ref{Sec_IntCSI_Out}, the outage probability for the case of NOMA is given by 
\begin{align}
	& \ \mathbb P_{\mathrm{out}} = 1 - \Pr \left( \dfrac{a_f P_t^* g_n}{a_n P_t^* g_n + 1} \geq \theta, a_n P_t^* g_n \geq \theta \right) \notag \\
	& \hspace{4cm}\times \Pr \left( \dfrac{a_f P_t^* g_f}{a_n P_t^* g_f + 1} \geq \theta \right) \notag \\
	= & \ 1 - \Pr \left( g_p \leq \tau \right) \Pr \left( g_n \geq \dfrac{\xi_n}{P_{\mathrm{peak}}}\right) \Pr \left( g_f \geq \dfrac{\xi_f}{P_{\mathrm{peak}}} \right) \notag \\
	= & \ 1 - \left\{\! 1\!  -\! \exp \left( \dfrac{-\tau}{\Omega_p}\right)\right\} \prod_{u \in \{n, f\}}  \dfrac{\Gamma[N_u, \xi_u/(\Omega_u P_{\mathrm{peak}})]}{\Gamma(N_u)}. \label{Sec_OneBit_OutClosed}
\end{align}

On the other hand, for the case of OMA, the outage probability is given by 
\begin{align}
	\mathscr P_{\mathrm{out}} = 1 - \prod_{u \in \{n, f\}} \Pr \left(g_u \geq \dfrac{\Theta}{P_t^*} \right), \label{Sec_OneBit_OutOMA}
\end{align}
where $P_t^*$ is given in~\eqref{Sec_OneBit_OptimalPower}.
\subsubsection*{Asymptotic analysis} For the limiting case $I \to \infty$, we have $(1 - \exp(-\tau/\Omega_p)) \to 1$. Therefore, using~\eqref{Sec_OneBit_OutClosed}, we have 
	\begin{align}
	\mathbb P_{\mathrm{out}} = 1 - \prod_{u \in \{n, f\}} \dfrac{\Gamma[N_n, \xi_u/(\Omega_u P_{\mathrm{peak}})]}{\Gamma(N_u)}, \label{Sec_OneBit_OutClosed_Asymptotic}
	\end{align}
which is the same as obtained for the case of PowIntICSI and PowIntSCSI. This occurs because for the limiting case of $I \to \infty$, the ST always receives a feedback `1' from the PT, and therefore, always transmits at a constant power $P_{\mathrm{peak}}$.

This means that for the case when the ST has a fixed power budget, the outage floor will always be the same for NOMA systems, regardless of the level of IL-CSI knowledge at the ST. 
\subsection{Optimal power allocation} \label{Sec_OneBit_OptPower}
Note that the expression for the outage probability for the NOMA system in the case of PowIntOneBit given in~\eqref{Sec_OneBit_OutClosed} is similar to that for the case of IntSCSI given in~\eqref{Sec_IntNoCSI_OutClosed}, where the difference is that $P_t^*$ in~\eqref{Sec_IntNoCSI_OutClosed} is replaced by $P_{\mathrm{peak}}$ in~\eqref{Sec_OneBit_OutClosed}, and there exists an extra factor of $(1 - \exp(-\tau/\Omega_p))$. Since both $P_{\mathrm{peak}}$ and $(1 - \exp(-\tau/\Omega_p))$ are independent of $a_n$, therefore, following the arguments in~Section~\ref{Sec_IntNoCSI_OptPower}, for the case where $N_n = N_f = 1$, $a_n^*$ is given by 
\begin{align}
		a_n^* = & \dfrac{\Omega_f}{(1 + \theta) \Omega_f - \Omega_n} - \dfrac{\sqrt{\Omega_n \Omega_f (1 + \theta)}}{(1 + \theta) \left\{(1 + \theta) \Omega_f -\Omega_n \right\}}. \label{Sec_OneBit_asStar}
\end{align}
Similar to case in~Section~\ref{Sec_IntNoCSI_OptPower}, $a_n^*$ does not depends on $P_{\mathrm{peak}}$, $I$, $\Omega_p$ or $\tau$.
\section{Results and Discussion} \label{Sec_Results}
In this section, we present the simulation and analytical results for the performance of the spectrum sharing NOMA/OMA systems. Throughout this section, we assume $d_n = 30$m, $d_f = 100$m, $d_p = 200$m, $\alpha = 2$, $\delta = 0.1$ and $N_n = N_f = N$. However, note that the analytical results presented in this paper are also valid for the case when $N_n \neq N_f$. 

Fig.~\ref{Fig_OutInt} shows a comparison between IntICSI and IntSCSI NOMA/OMA systems in terms of outage probability, for different values of $N$. It is clear from the figure that the spectrum sharing NOMA system significantly outperforms the corresponding OMA system for both IntICSI and IntSCSI cases. It is important to note that for large values of $I$, the difference between the outage probability of the NOMA system and the corresponding OMA system increases with an increase in the value of $N$. It is also noteworthy that for large values of $I$, the difference between the outage probability of the NOMA system for IntICSI and IntSCSI decreases with an increase in the value of $N$, indicating that the impact of information loss becomes less significant for larger values of $N$ and $I$. 

Fig.~\ref{Fig_OutPowIntCSI} shows the outage probability of the PowIntICSI system for both NOMA and OMA, with different values of $N$ and $P_{\mathrm{peak}}$. For both NOMA and OMA systems, the outage probability first decreases for small values of $I$ (which we refer to as the \emph{interference-constrained} regime) and then becomes saturated for large values of $I$ (which we refer to as the \emph{power-constrained} regime), as is evident from~\eqref{Sec_PowIntCSI_OutClosed_Asymptotic}. This occurs because the average power transmitted from the ST first increases with an increase in the value of $I$ and when the value of $I$ is large, the average power transmitted from the ST becomes constant, resulting in an outage floor. It is evident from the figure that the outage probability of NOMA system is significantly lower than that of the corresponding OMA system. More interestingly, for the NOMA/OMA system, the outage probability remains (almost) the same in the interference-constrained regime for a fixed value of $N$, regardless of the value of $P_{\mathrm{peak}}$, whereas the effect of $P_{\mathrm{peak}}$ becomes significantly evident in the power-constrained regime. 

\begin{figure*}[t]
\centering
\begin{minipage}{.32\textwidth}
  \centering
  \includegraphics[width = 0.9\linewidth, height = 4cm]{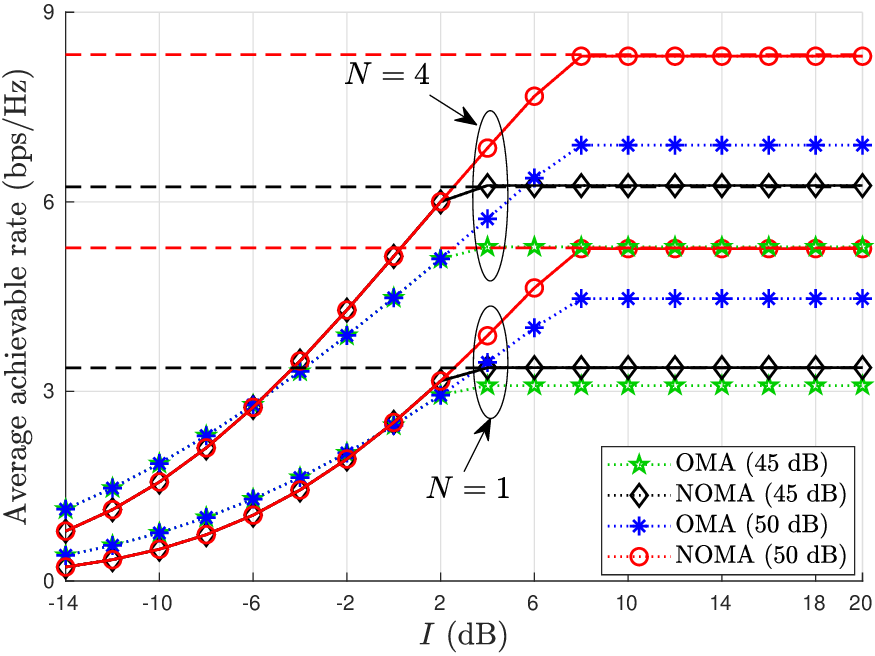}
  \caption{Average achievable sum-rate for the PowIntSCSI system. Markers with dotted lines represent the simulation results, the solid lines drawn using~\eqref{Sec_IntNoCSI_CapClosed} with $P_t^*$ given by~\eqref{Sec_PowIntNoCSI_PtOptimal} represent the analytical results, and the horizontal dotted lines drawn using~\eqref{Sec_IntNoCSI_CapClosed}, with $P_t^*$ replaced by $P_{\mathrm{peak}}$, denote the asymptotic results (for NOMA).}
  \label{Fig_CapPowIntNoCSI}
\end{minipage}
\hfill 
\begin{minipage}{.32\textwidth}
  \centering
  \includegraphics[width = 0.9\linewidth, height = 4cm]{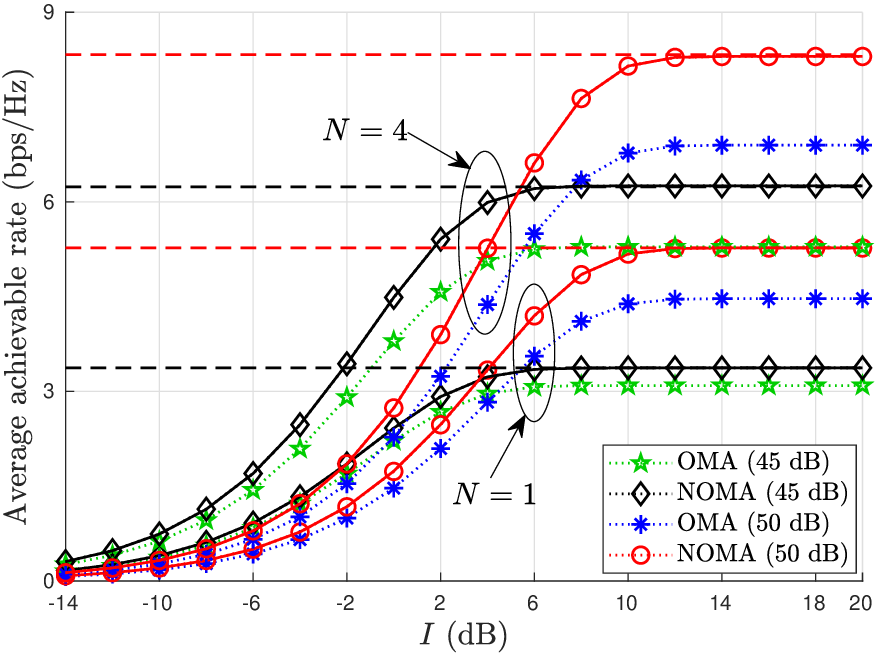}
  \caption{Average achievable sum-rate for the PowIntOneBit system. Markers with dotted lines represent the simulation results, the solid lines drawn using~\eqref{Sec_OneBit_CapClosed} denote the analytical results, and the horizontal dotted lines drawn using~\eqref{Sec_IntNoCSI_CapClosed}, with $P_t^*$ replaced by $P_{\mathrm{peak}}$, represent asymptotic results (for NOMA).}
  \label{Fig_CapOneBit}
\end{minipage}
\hfill
\begin{minipage}{.32\textwidth}
	\centering
	\includegraphics[width = 0.9\linewidth, height = 4cm]{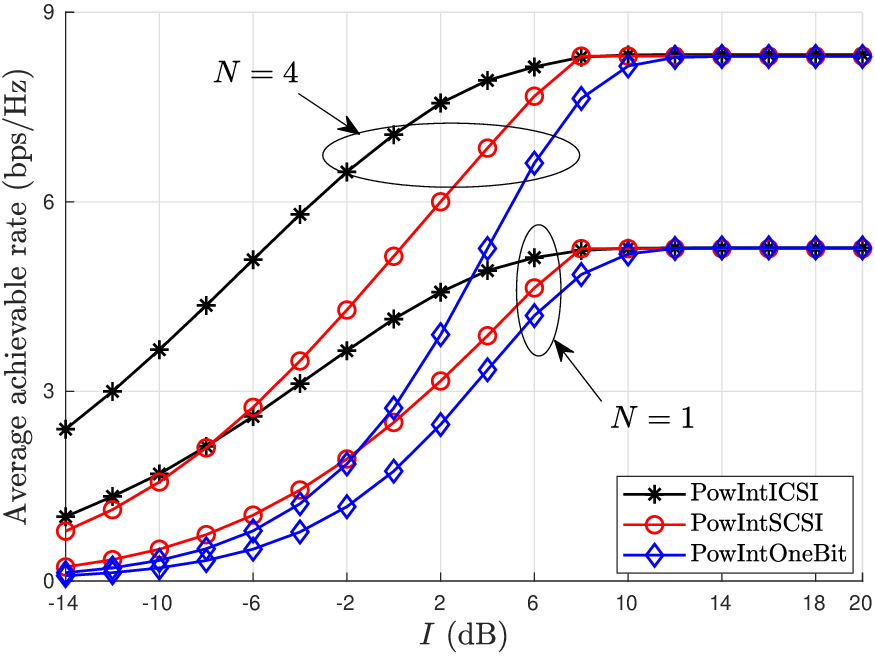}
\caption{Analytical plots for the average achievable sum-rate of the PowIntICSI, PowIntSCSI and PowIntOneBit NOMA systems. Here $P_{\mathrm{peak}}$ is fixed at 50 dB.}
\label{Fig_CapPow}
\end{minipage}
\end{figure*}

The outage probability of the PowIntSCSI NOMA/OMA system is shown in Fig.~\ref{Fig_OutPowIntNoCSI} for different values of $N$ and $P_{\mathrm{peak}}$. Similar to the case of Fig.~\ref{Fig_OutPowIntCSI}, the interference-constrained regime and power-constrained regime (see~\eqref{Sec_PowIntNoCSI_OutClosed_Asymptotic}) are clearly evident in Fig.~\ref{Fig_OutPowIntNoCSI}, with the NOMA system outperforming the corresponding OMA system. However, different from the case in Fig.~\ref{Fig_OutPowIntCSI}, the outage probability of NOMA/OMA for a fixed value of $N$ is \emph{exactly} the same in the interference-constrained regime, irrespective of the value of $P_{\mathrm{peak}}$. This occurs because the power transmitted by the ST is given by $P_t^* = \min\{P_{\mathrm{peak}}, -I/(\Omega_p \ln \delta)\}$; in the interference-constrained regime, this is equal to $- I / (\Omega_p \ln \delta)$, which is independent of $P_{\mathrm{peak}}$. 

Fig.~\ref{Fig_OutOneBit} depicts the outage probability performance of the PowIntOneBit NOMA/OMA systems for different values of $N$ and $P_{\mathrm{peak}}$. It is clearly evident from the figure that the NOMA system outperforms its OMA-based counterpart, by achieving a lower outage probability. However, different from the case in Figs.~\ref{Fig_OutPowIntCSI} and~\ref{Fig_OutPowIntNoCSI}, for a fixed value of $N$, the outage probability of NOMA/OMA systems with larger value of $P_{\mathrm{peak}}$ is higher in the interference-constrained regime. This occurs because when the value of $P_{\mathrm{peak}}$ is large (for a fixed $N$), the value of $\tau \triangleq I/P_{\mathrm{peak}}$ becomes small and therefore, the probability of receiving a feedback ``1'' at the ST becomes smaller (c.f.~\eqref{Sec_OneBit_OptimalPower}), which in turn leads to a higher probability of the ST being silent. Therefore, different from the other cases, having a higher peak power budget is not always beneficial in the interference-constrained regime for the case of the PowIntOneBit system. However, in the power-constrained regime (see~\eqref{Sec_OneBit_OutClosed_Asymptotic}), having a large $P_{\mathrm{peak}}$ is always advantageous, as is evident from the figure.


Fig.~\ref{Fig_OutPow} shows a comparison of the outage probability for PowIntICSI, PowIntSCSI and PowIntOneBit NOMA systems with $P_{\mathrm{peak}} = 50$ dB.  It is evident from the figure that in the power-constrained regime, the outage probability for all the three NOMA systems for a fixed value of $N$ converges to the same outage floor, in accordance with the asymptotic analysis of the outage probability of the corresponding systems. In the interference-constrained regime, the effect of information loss in terms of IL-CSI between PowIntICSI and PowIntSCSI systems is not very significant. However, in the case of PowIntOneBit system, the effect of information loss in terms of IL-CSI becomes significantly dominant in the interference-constrained regime, especially for large $N$.

Fig.~\ref{Fig_CapInt} shows a comparison of the average achievable sum-rate for IntICSI and IntSCSI NOMA and OMA systems. It is evident from the figure that the NOMA-based system outperforms the corresponding OMA-based system in terms of achievable sum-rate for large values of $N$. It is noteworthy that in contrast to the behavior in the case of outage probability, the performance degradation in IntSCSI system as compared to IntICSI system in terms of achievable rate due to information loss is significant even for large values of $N$. We have also plotted the asymptotic sum-rate for the NOMA system using~\eqref{Sec_IntCSI_CapClosed_Asymptotic} for IntICSI and~\eqref{Sec_IntNoCSI_CapClosed_Asymptotic} for IntSCSI, which shows a close agreement with the exact sum-rate for large $I$. 

In Fig.~\ref{Fig_CapPowIntCSI}, the average achievable sum-rate for the PowIntICSI NOMA/OMA system is shown for different values of $N$ and $P_{\mathrm{peak}}$. Interestingly, the difference between the sum-rate of the NOMA systems with $P_{\mathrm{peak}} = 45$ dB and $50$ dB is less significant in the interference-constrained regime, whereas the performance difference between the two systems becomes significant in the power-constrained regime. 

Fig.~\ref{Fig_CapPowIntNoCSI} shows the average achievable sum-rate for PowIntSCSI NOMA and OMA systems for different values of $N$ and $P_{\mathrm{peak}}$. It is noteworthy from the figure that for a fixed $N$ in the interference-constrained regime, the sum-rate for the NOMA/OMA system is \emph{exactly} the same for both $P_{\mathrm{peak}} = 45$ dB and $50$ dB, because of the same reason as explained previously for~Fig.~\ref{Fig_OutPowIntNoCSI}. Also, as explained for the previous figure, the achievable rate for both NOMA and OMA systems saturates in the power-constrained regime (see~\eqref{Sec_IntNoCSI_CapClosed}, with $P_t^*$ replaced by $P_{\mathrm{peak}}$) for NOMA, due to the fact that the power transmitted from the ST is constant and independent of the value of $I$.

Fig.~\ref{Fig_CapOneBit} depicts the average achievable sum-rate performance of the PowIntOneBit NOMA and OMA systems for different values of $N$ and $P_{\mathrm{peak}}$. For a fixed value of $N$ in the interference-constrained regime, the sum-rate of the NOMA/OMA system with larger $P_{\mathrm{peak}}$ achieves lower sum-rate as compared to the NOMA/OMA system with smaller $P_{\mathrm{peak}}$, because of the same reason as explained previously for~Fig.~\ref{Fig_OutOneBit}. Therefore, having a higher peak power budget in the PowIntOneBit NOMA/OMA system is not always beneficial in the interference-constrained regime. However, in the power-constrained regime, a larger value of $P_{\mathrm{peak}}$ results in a higher achievable sum-rate.  

Fig.~\ref{Fig_CapPow} shows the achievable sum-rate performance of PowIntICSI, PowIntSCSI and PowIntOneBit NOMA systems for the case when $P_{\mathrm{peak}} = 50$ dB. It can be noticed from the figure that in the interference-constrained regime, there is a significant performance degradation due to the information loss in terms of IL-CSI, whereas in the power-constrained regime, there is no effect of information loss in terms of IL-CSI. 
\section{Conclusion} \label{Sec_Conclusion}
In this paper, we presented the performance analysis of a multi-antenna-assisted NOMA-based underlay spectrum sharing system over Rayleigh fading channels. We derived closed-form expressions for the average achievable sum-rate and outage probability for the downlink NOMA system under a peak interference as well as a peak power budget constraint with different CSI availability at the ST regarding the link between ST and PR. Our results confirm that for a large number of antennas at the secondary users, the performance difference between the system with instantaneous IL-CSI and statistical IL-CSI in the interference-constrained regime is negligible in terms of outage probability, whereas this difference is significant in terms of achievable sum-rate. On the other hand when no IL-CSI is available at the ST, the NOMA and OMA systems both suffer from a significant performance degradation in the interference-constrained regime for a large number of antennas, in terms of outage probability as well as achievable sum-rate. However, in the power-constrained regime, the effect of information loss in IL-CSI is negligible for both outage probability and achievable sum-rate. We also derived closed-form expressions for the optimal power allocation to minimize the outage probability of NOMA systems for the special case when the secondary users are each equipped with a single antenna.
\appendices
\section{Proof of Theorem~\ref{Sec_IntCSI_TheoremCap}} \label{Sec_IntCSI_ProofCap}
Given that channel gains for all of the wireless links are exponential distributed, the PDF and CDF of $g_u (u \in \{n, f\})$, are respectively given by 
\begin{align}
	f_{g_u}(x) = \dfrac{x^{N_u - 1}}{\Gamma(N_u) \Omega_u^{N_u}} \exp \left( \dfrac{-x}{\Omega_u}\right), \label{f_gu} 
\end{align}
and 
\begin{align}
	F_{g_u}(x) = 1 - \exp \left( \dfrac{-x}{\Omega_u} \right) \sum_{k = 0}^{N_u - 1} \dfrac{1}{k!} \left( \dfrac{x}{\Omega_u} \right)^{k}. \label{F_gu}
\end{align}
Also, the PDF of $g_p$ is given by 
\begin{align}
	f_{g_p}(x) = \dfrac{1}{\Omega_p} \exp\left( \dfrac{-x}{\Omega_p}\right). \label{f_gp}
\end{align}
Now, the PDF of $g_{\min}$ can be obtained by
\begin{align}
	f_{g_{\min}}(x) = & \ f_{g_n}(x) [1 - F_{g_f}(x)] + f_{g_f}(x) [1 - F_{g_n}(x)] \notag \\
	= & \ \dfrac{1}{\Gamma(N_n) \Omega_n^{N_n}} \sum_{k = 0}^{N_f - 1} \dfrac{x^{N_n + k - 1}}{k! \Omega_f^k} \exp \left( \dfrac{-x}{\Omega} \right) \notag \\
	& + \dfrac{1}{\Gamma(N_f) \Omega_f^{N_f}} \sum_{l = 0}^{N_n - 1} \dfrac{x^{N_f + l - 1}}{l! \Omega_n^l} \exp \left( \dfrac{-x}{\Omega} \right). \label{f_gmin}
\end{align}
Also, the PDF of $X_u \triangleq g_u/g_p$ can be obtained as 
\begin{align}
	& f_{X_u}(x) = \int_0^\infty y f_{g_u} (yx) f_{g_p} (y) \mathrm dy \notag \\
	= & \ \dfrac{x^{N_u - 1}}{\Gamma(N_u) \Omega_u^{N_u} \Omega_p} \int_0^\infty y^{N_u} \exp \left[ -\left( \dfrac{x}{\Omega_u} + \dfrac{1}{\Omega_p} \right) y \right] \mathrm dy \notag \\
	= & \ \dfrac{N_u x^{N_u - 1}}{\Omega_u^{N_u} \Omega_p} \left( \dfrac{x}{\Omega_u} + \dfrac{1}{\Omega_p} \right)^{-(N_u + 1)}, \label{f_Xu}
\end{align}
where the integration above is solved using~\cite[eqn.~(3.351-3),~p.~340]{Grad}. Similarly, the PDF of $X_{\min} \triangleq g_{\min}/g_p$ can be given by 
\begin{align}
	& f_{X_{\min}}(x) = \int_0^\infty y f_{g_{\min}}(yx) f_{g_p}(y) \mathrm dy \notag \\
	= & \ \dfrac{1}{\Gamma(N_n) \Omega_n^{N_n} \Omega_p} \sum_{k = 0}^{N_f - 1} \dfrac{x^{N_n + k - 1}}{k! \Omega_f^k} \int_0^\infty y^{N_n + k}  \notag \\
	& \times \exp \left[ -\left( \dfrac{x}{\Omega} + \dfrac{1}{\Omega_p} \right) y \right] \mathrm dy + \dfrac{1}{\Gamma(N_f) \Omega_f^{N_f} \Omega_p} \notag \\
	& \times \sum_{l = 0}^{N_n - 1} \dfrac{x^{N_f + l - 1}}{l! \Omega_n^l} \int_0^\infty y^{N_f + l} \exp \left[ -\left( \dfrac{x}{\Omega} + \dfrac{1}{\Omega_p} \right) y \right] \mathrm dy \notag \\
	= & \sum_{k = 0}^{N_f - 1} \dfrac{x^{N_n + k - 1} \Gamma(N_n + k + 1)}{\Gamma(N_n) k! \Omega_n^{N_n} \Omega_f^k \Omega_p} \left( \dfrac{x}{\Omega} + \dfrac{1}{\Omega_p} \right)^{-(N_n + k + 1)} \notag \\
	& + \sum_{l = 0}^{N_n- 1} \dfrac{x^{N_f + l - 1} \Gamma(N_f + l + 1)}{\Gamma(N_f) l! \Omega_f^{N_f} \Omega_n^l \Omega_p} \left( \dfrac{x}{\Omega} + \dfrac{1}{\Omega_p} \right)^{-(N_f + l + 1)}. \label{f_Xmin}
\end{align}
The integration above is solved using~\cite[eqn.~(3.351-3),~p.~340]{Grad}. Using~\eqref{f_Xu}, an analytical expression for the first expectation in~\eqref{Sec_IntCSI_Csum_Def} is be given by  
\begin{align}
	& \mathbb E_{X_n} \{ \log_2 (1 + a_n I X_n)\} = \dfrac{1}{\ln 2} \int_0^\infty \ln(1 + a_n I x) f_{X_n}(x) \mathrm dx \notag \\
	 = & \dfrac{N_n \Omega_p^{N_n}}{\Omega_n^{N_n} \ln 2} \int_0^\infty x^{N_n - 1} \ln(1 + a_n I x) \left( 1 + \dfrac{\Omega_p}{\Omega_n} x \right)^{-(N_n + 1)} \mathrm dx \notag \\
	 = & \dfrac{N_n \Omega_p^{N_n}}{\Omega_n^{N_n} \Gamma(N_n + 1)\ln 2} \int_0^\infty x^{N_n - 1} G_{2, 2}^{1, 2} \left( a_n I x \left \vert \begin{smallmatrix} 1, & 1 \\[0.6em] 1, & 0\end{smallmatrix} \right. \right)  \notag \\
	 & \hspace{4.5cm}\times G_{1, 1}^{1, 1} \left( \dfrac{\Omega_p}{\Omega_n} x \left \vert \begin{smallmatrix} -N_n \\[0.6em] 0\end{smallmatrix} \right. \right) \mathrm dx \notag \\
	 = & \ \dfrac{1}{\Gamma(N_n) \ln 2} \left( \dfrac{\Omega_p}{\Omega_n a_n I} \right)^{N_n}  \!\! G_{3, 3}^{3, 2} \left( \dfrac{\Omega_p}{\Omega_n a_n I} \left \vert \begin{smallmatrix} -N_n, & - N_n, & 1 - N_n \\[0.6em] 0, & - N_n, & - N_n\end{smallmatrix} \right. \right), \label{Theorem1-FirstIntegralClosed}
\end{align}
where the integration above is solved using~\cite[eqns.~(7),~(11),~(21), and~(22)]{Reduce}. Similarly, using~\eqref{f_Xmin}, an analytical expression for the second expectation in~\eqref{Sec_IntCSI_Csum_Def} is given by 
\begin{align}
	& \mathbb E_{X_{\min}}\{\log_2 (1 + I X_{\min})\} \notag \\
	= & \ \dfrac{1}{\ln 2} \left[ \sum_{k = 0}^{N_f - 1} \dfrac{\Gamma(N_n + k + 1)}{\Gamma(N_n) k! \Omega_n^{N_n} \Omega_f^k \Omega_p} \int_0^\infty x^{N_n + k - 1} \ln (1 + Ix) \right. \notag \\
	& \times \left( \dfrac{x}{\Omega} + \dfrac{1}{\Omega_p} \right)^{-(N_n + k + 1)} \mathrm dx + \sum_{l = 0}^{N_n - 1} \dfrac{\Gamma(N_f + l + 1)}{\Gamma(N_f) l! \Omega_f^{N_f} \Omega_n^{l} \Omega_p} \notag \\
	& \left. \times \int_0^\infty x^{N_f + l - 1} \ln(1 + Ix) \left( \dfrac{x}{\Omega} + \dfrac{1}{\Omega_p}\right)^{-(N_f + l + 1)} \mathrm dx \right] \notag  \\
	= & \dfrac{1}{\ln 2} \left[ \sum_{k = 0}^{N_f - 1} \dfrac{\Omega_p^{N_n + k}}{\Gamma(N_n) k! \Omega_n^{N_n} \Omega_f^k} \int_0^\infty x^{N_n + k - 1} G_{2, 2}^{1, 2} \left( I x \left \vert \begin{smallmatrix} 1, & 1 \\[0.6em] 1, & 0\end{smallmatrix} \right. \right) \right. \notag \\
	& \times G_{1, 1}^{1, 1} \left( \dfrac{\Omega_p}{\Omega} x \left \vert \begin{smallmatrix} - N_n - k\\[0.6em] 0\end{smallmatrix} \right. \right) \mathrm dx + \sum_{l = 0}^{N_n - 1} \dfrac{\Omega_p^{N_f + l}}{\Gamma(N_f) l! \Omega_f^{N_f} \Omega_n^l}   \notag \\
	& \left. \times \int_0^\infty x^{N_f + l - 1} G_{2, 2}^{1, 2} \left( I x \left \vert \begin{smallmatrix} 1, & 1 \\[0.6em] 1, & 0\end{smallmatrix} \right. \right) G_{1, 1}^{1, 1} \left( \dfrac{\Omega_p}{\Omega} x \left \vert \begin{smallmatrix} - N_f - l\\[0.6em] 0\end{smallmatrix} \right. \right) \mathrm dx \right] \notag \\
	= & \ \dfrac{1}{\ln 2} \left[ \sum_{k = 0}^{N_f - 1} \dfrac{\Omega_p^{N_n + k} G_{3, 3}^{3, 2} \left( \dfrac{\Omega_p}{\Omega I} \left \vert \begin{smallmatrix} - N_n - k, & - N_n - k, & 1 - N_n - k \\[0.6em] 0, & - N_n - k, & - N_n - k\end{smallmatrix} \right. \right)}{\Gamma(N_n) k! \Omega_n^{N_n} \Omega_f^k I^{N_n + k}} \right. \notag \\
	& \left. + \sum_{l = 0}^{N_n - 1} \dfrac{\Omega_p^{N_f + l} G_{3, 3}^{3, 2} \left( \dfrac{\Omega_p}{\Omega I} \left \vert \begin{smallmatrix} - N_f - l, & - N_f - l, & 1 - N_f - l \\[0.6em] 0, & - N_f - l, & - N_f - l\end{smallmatrix} \right. \right)}{\Gamma(N_f) l! \Omega_f^{N_f} \Omega_n^l I^{N_f + l}}  \right]. \label{Theorem1-SecondIntegralClosed}
\end{align}
The integral above is solved in a similar fashion as in~\eqref{Theorem1-FirstIntegralClosed}. An analytical expression for the third expectation in~\eqref{Sec_IntCSI_Csum_Def} can be obtained by replacing $I$ with $a_n I$ in~\eqref{Theorem1-SecondIntegralClosed}. Therefore, using~\eqref{Theorem1-FirstIntegralClosed} and~\eqref{Theorem1-SecondIntegralClosed}, an analytical expression for~\eqref{Sec_IntCSI_Csum_Def} is given by~\eqref{Sec_IntCSI_CapClosed}; this concludes the proof.
\section{Proof of Theorem~\ref{Sec_IntCSI_TheoremOut}} \label{Sec_IntCSI_ProofOut}
We first define the \emph{non-outage} event for NOMA as the event where $x_f$ and $x_n$ are decoded successfully at $U_n$, and $x_f$ is decoded successfully at $U_f$. Therefore, the outage probability for the NOMA system is given by 
\begin{align}
	& \mathbb P_{\mathrm{out}} = 1 - \Pr\left( \gamma_n^{(f)} \geq \theta, \gamma_n^{(n)} \geq \theta\right) \Pr \left( \gamma_f^{(f)} \geq \theta \right) \notag \\
	= & \ 1 - \Pr \left( \dfrac{a_f I g_n/g_p}{1 + a_n I g_n/g_p} \geq \theta, \dfrac{a_n I g_n}{g_p} \geq \theta \right) \notag \\ 
	& \hspace{4cm}\times \Pr \left(\dfrac{a_f I g_f/g_p}{1 + a_n I g_f/g_p} \geq \theta \right)\notag \\
	= & \ 1 - \Pr \left( \dfrac{a_f I X_n}{1 + a_n I X_n} \geq \theta, \ a_n I X_n \geq \theta \right) \notag \\
	& \hspace{4cm}\times \Pr \left(\dfrac{a_f I X_f}{1 + a_n I X_f} \geq \theta \right) \notag \\
	= & \ 1 - \Pr \left( X_n \geq \dfrac{\theta}{I} \max \left\{ \dfrac{1}{a_f - a_n \theta}, \dfrac{1}{a_n}\right\} \right) \notag \\
	& \hspace{4cm}\times \Pr \left( X_f \geq \dfrac{\theta}{I (a_f - a_n \theta)}\right). \notag
\end{align}
Using the relations $\xi_n = \theta \max \left\{ \tfrac{1}{a_f - a_n \theta}, \tfrac{1}{a_n}\right\}$	and $\xi_f = \tfrac{\theta}{a_f - a_n \theta}$, the expression for $\mathbb P_{\mathrm{out}}$ can be written as
\begin{align}	
	\mathbb P_{\mathrm{out}} = & \ 1 - \Pr\left(X_n \geq \dfrac{\xi_n}{I}\right) \Pr\left(X_f \geq \dfrac{\xi_f}{I}\right) \notag \\
	= & \ 1 - \prod_{u \in \{n, f\}}\mathcal F_{X_u} \left(\dfrac{\xi_u}{I}\right), \label{Sec_IntCSI_OutIntermediate}
\end{align}
where $\mathcal F_{X_u}(\cdot)$ is evaluated as 
\begin{align}
	\mathcal F_{X_u}\left(\dfrac{\xi_u}{I}\right) = & \ \int_{\xi_u/I}^{\infty} f_{X_u}(x) \mathrm dx \notag \\
	= & \ \dfrac{N_u \Omega_u}{\Omega_p} \int_{\xi_u/I}^\infty  x^{N_u - 1} \left( x + \dfrac{\Omega_u}{\Omega_p} \right)^{-(N_u + 1)} \mathrm dx \notag \\
	 = & \ 1- \left( \dfrac{\Omega_p \xi_u}{\Omega_u I + \Omega_p \xi_u}\right)^{N_u}. \label{mathcalFu}
\end{align}
Substituting the expression for $\mathcal F_{X_u}(\xi_u/I)$ from~\eqref{mathcalFu} into~\eqref{Sec_IntCSI_OutIntermediate}, the closed-form expression for $\mathbb P_{\mathrm{out}}$ becomes equal to~\eqref{Sec_IntCSI_OutClosed}; this concludes the proof.
\section{Proof of Theorem~\ref{Sec_IntCSI_TheoremOpt}} \label{Sec_IntCSI_ProofOpt}
For the case where $N_n = N_f = 1$, the outage probability is given by 
\begin{align}
	\mathbb P_{\mathrm{out}} = 1 - \prod_{u \in \{n, f\}} \dfrac{\Omega_u}{\Omega_u + \Omega_p (\xi_u/I)}. \notag
\end{align}
Assuming, $1/(a_f - a_n \theta) > 1/a_n$, i.e., $\xi_n = 1/(1- a_n - a_n \theta)$, we have 
\begin{align*}
	\dfrac{\partial \mathbb P_{\mathrm{out}}}{\partial a_n} = & \ \dfrac{I^2 \Omega_p \Omega_n \Omega_f \theta (1 + \theta) (a_n + a_n \theta - 1)}{\{\Omega_p \theta - I \Omega_n (a_n + a_n \theta - 1)\}^2} \\
	& \times \dfrac{\{I (\Omega_n + \Omega_f) (a_n + a_n \theta - 1) - 2 \Omega_p \theta\}}{\{\Omega_p \theta - I \Omega_f (a_n + a_n \theta - 1)\}^2}.
\end{align*}
Using the fact that $a_n < 1/(1 + \theta)$ (see~Section~\ref{Sec_IntCSI_Out}), we have 
\begin{align*}
	& \dfrac{\partial \mathbb P_{\mathrm{out}}}{\partial a_n} = 0 \\
	\implies & \ I (\Omega_f + \Omega_n) \left\{ (1 + \theta) a_n - 1\right\}^2 \\
	& \hspace{1.5cm}- 2\theta \Omega_p \left\{ (1 + \theta) a_n - 1\right\} = 0.
\end{align*}
The preceding equation is quadratic, leading to the following two solutions:
\begin{align*}
	a_n = \dfrac{1}{1 + \theta}, \quad \dfrac{1}{1 + \theta} + \dfrac{2 \Omega_p \theta}{I (\Omega_n + \Omega_f) (1 + \theta)}.
\end{align*}
Since, $a_n < 1/(1 + \theta)$, neither of the above optimal values is feasible. Now assuming that $1/(a_f - a_n \theta) < 1/a_n$, i.e., $\xi_n = 1/a_n$, we have 
\begin{align*}
	\dfrac{\partial \mathbb P_{\mathrm{out}}}{\partial a_n} = & \ \dfrac{I^2 \Omega_p \Omega_n \Omega_f \theta}{(a_n I \Omega_n + \Omega_p \theta)^2 \{\Omega_p \theta - I \Omega_f (a_n + a_n \theta - 1)\}^2} \\
	& \times \left\{ -I \Omega_f -\Omega_p \theta + 2 a_n (1 + \theta) (I \Omega_f + \Omega_p \theta)  \right. \\
	& \hspace{2cm}\left. - a_n^2 I (1 + \theta) (-\Omega_n + \Omega_f + \Omega_f \theta)\right\}
\end{align*}
Since $a_n < 1/(1 + \theta)$, we have 
\begin{align*}
	& \dfrac{\partial \mathbb P_{\mathrm{out}}}{\partial a_n} = 0 \\
	\implies & \ I^2 \Omega_p \Omega_n \Omega_f \theta \left\{ -I \Omega_f -\Omega_p \theta + 2 a_n (1 + \theta) (I \Omega_f + \Omega_p \theta) \right. \\
	& \hspace{2cm}\left. - a_n^2 I (1 + \theta) (-\Omega_n + \Omega_f + \Omega_f \theta)\right\} = 0 \\
	\implies & \ a_n^2 I (1 + \theta) (-\Omega_n + \Omega_f + \Omega_f \theta) \\
	& \hspace{1.3cm}- 2 a_n (1 + \theta) (I \Omega_f + \Omega_p \theta) + I \Omega_f +\Omega_p \theta = 0 \\
	\implies & a_n = \ \dfrac{I \Omega_f + \Omega_p \theta}{I \{(1 + \theta) \Omega_f - \Omega_n\}} \notag \\
		& \pm \dfrac{\sqrt{(1 + \theta) (I \Omega_f + \Omega_p \theta) \{I \Omega_n + \Omega_p \theta (1 + \theta)\}}}{I (1 + \theta) \{(1 + \theta) \Omega_f - \Omega_n\}}.
\end{align*}
Define
\begin{align*}
	a_n^{(1)} \triangleq & \ \dfrac{I \Omega_f + \Omega_p \theta}{I \{(1 + \theta) \Omega_f - \Omega_n\}} \notag \\
		& + \dfrac{\sqrt{(1 + \theta) (I \Omega_f + \Omega_p \theta) \{I \Omega_n + \Omega_p \theta (1 + \theta)\}}}{I (1 + \theta) \{(1 + \theta) \Omega_f - \Omega_n\}} \\
		= & \dfrac{1}{(1 + \theta) - (\Omega_n/\Omega_f)} + \dfrac{\Omega_p \theta / (I \Omega_f)}{(1 + \theta) - (\Omega_n/\Omega_f)} \\
		& + \dfrac{\sqrt{\tfrac{(I \Omega_f + \Omega_p \theta)}{I^2 (1 + \theta) \Omega_f^2}\{I \Omega_n + \Omega_p \theta (1 + \theta)\}}}{(1 + \theta) - (\Omega_n/\Omega_f)} \\
		\triangleq & \mathscr X_1 + \mathscr X_2 + \mathscr X_3.
\end{align*}
For the case $1 + \theta > \Omega_n / \Omega_f$, $\mathscr X_1, \mathscr X_2, \mathscr X_3 > 0$ and $\mathscr X_1 > 1/(1 + \theta)$. Therefore, $a_n^{(1)} > 1/(1 + \theta)$. On the other hand, if $1 + \theta < \Omega_n / \Omega_f$, $\mathscr X_1, \mathscr X_2, \mathscr X_3 < 0$ and $a_n^{(1)} < 0$. Therefore, the only feasible solution for the optimal value of $a_n$ is 
\begin{align*}
	a_n^* = a_n^{(2)} = & \ \dfrac{I \Omega_f + \Omega_p \theta}{I \{(1 + \theta) \Omega_f - \Omega_n\}} \notag \\
		& - \dfrac{\sqrt{(1 + \theta) (I \Omega_f + \Omega_p \theta) \{I \Omega_n + \Omega_p \theta (1 + \theta)\}}}{I (1 + \theta) \{(1 + \theta) \Omega_f - \Omega_n\}}.
\end{align*}
This concludes the proof.
\section{Proof of Theorem~\ref{Sec_IntNoCSI_TheoremCap}} \label{Sec_IntNoCSI_ProofCap}
Using~\eqref{f_gu}, an analytical expression for the first expectation in~\eqref{Sec_IntNoCSI_Cap_Def} can be given by 
\begin{align}
	& \mathbb E_{g_n} \left\{ \log_2 (1 + a_n g_n P_t^*)\right\} \notag \\
	= & \ \dfrac{1}{\Gamma(N_n) \Omega_n^{N_n} \ln 2} \int_0^\infty x^{N_n - 1} \ln(1 + a_n P_t^* x) \exp \left( \dfrac{-x}{\Omega_n}\right) \mathrm dx \notag \\ 
	= & \dfrac{1}{\Gamma(N_n) \Omega_n^{N_n} \ln 2} \int_0^\infty x^{N_n - 1} G_{2, 2}^{1, 2} \left( a_n P_t^* x \left\vert \begin{smallmatrix} 1, & 1 \\ 1, & 0\end{smallmatrix} \right. \right) \notag \\
	& \hspace{5.3cm} \times G_{0, 1}^{1, 0} \left( \dfrac{x}{\Omega_n} \left\vert \begin{smallmatrix} - \\ 0\end{smallmatrix} \right. \right) \mathrm dx \notag \\
	= & \ \dfrac{G_{2, 3}^{3, 1} \left( \dfrac{1}{\Omega_n a_n P_t^*} \left\vert \begin{smallmatrix} -N_n, & 1 - N_n \\[0.6em] 0, & - N_n, & - N_n \end{smallmatrix} \right. \right)}{\Gamma(N_n) \Omega_n^{N_n} (a_n P_t^*)^{N_n} \ln 2}. \label{Theorem3-FirstIntegralClosed}
\end{align}
The integral above is solved using~\cite[eqns.~(7),~(11),~(21), and~(22)]{Reduce}. Using~\eqref{f_gmin}, an analytical expression for the second expectation in~\eqref{Sec_IntNoCSI_Cap_Def} can be given by 
\begin{align}
	& \mathbb E_{g_{\min}} \left\{ \log_2 (1 + g_{\min} P_t^*) \right\} \notag \\
	= & \ \dfrac{1}{\Gamma(N_n) \Omega_n^{N_n} \ln 2} \sum_{k = 0}^{N_f - 1} \dfrac{1}{k! \Omega_f^k} \int_0^\infty x^{N_n + k - 1} \ln (1 + P_t^* x) \notag \\
	& \times \exp \left( \dfrac{-x}{\Omega} \right)  \mathrm dx + \dfrac{1}{\Gamma(N_f) \Omega_f^{N_f} \ln 2} \sum_{l = 0}^{N_n - 1} \dfrac{1}{l! \Omega_n^l} \notag \\
	& \hspace{2cm}\times \int_0^\infty x^{N_f + l - 1} \ln(1 + P_t^* x) \exp \left( \dfrac{-x}{ \Omega} \right) \mathrm dx \notag \\
	= & \ \dfrac{1}{\Gamma(N_n) \Omega_n^{N_n} \ln 2} \sum_{k = 0}^{N_f - 1} \dfrac{G_{2, 3}^{3, 1} \left( \dfrac{1}{\Omega  P_t^*} \left\vert \begin{smallmatrix} -N_n-k, 1 - N_n -k \\[0.6em] 0, - N_n -k,  - N_n -k \end{smallmatrix} \right. \right)}{k! \Omega_f^k (P_t^*)^{N_n + k}} \notag \\
	+ & \dfrac{1}{\Gamma(N_f) \Omega_f^{N_f} \ln 2} \times \sum_{l = 0}^{N_n - 1} \dfrac{G_{2, 3}^{3, 1} \left( \dfrac{1}{\Omega  P_t^*} \left\vert \begin{smallmatrix} -N_f-l, 1 - N_f -l \\[0.6em] 0, - N_f -l, - N_f - l \end{smallmatrix} \right. \right)}{l! \Omega_n^l (P_t^*)^{N_f + l}}. \label{Theorem3-SecondIntegralClosed} 
\end{align}
The integrals above are solved using~\cite[eqns.~(7),~(11),~(21), and~(22)]{Reduce}. An analytical expression for the third expectation in~\eqref{Sec_IntNoCSI_Cap_Def} can be obtained by replacing $P_t^*$ by $a_n P_t^*$ in~\eqref{Theorem3-SecondIntegralClosed}. Therefore, using~\eqref{Theorem3-FirstIntegralClosed} and~\eqref{Theorem3-SecondIntegralClosed}, an analytical expression for~\eqref{Sec_IntNoCSI_Cap_Def} is given by~\eqref{Sec_IntNoCSI_CapClosed}; this completes the proof.
\section{Proof of Theorem~\ref{Sec_IntNoCSI_TheoremOpt}} \label{Sec_IntNoCSI_ProofOpt}
Using the relation $\Gamma[1, x] = \exp(-x)$, the expression for the outage probability for the case when $N_n = N_f = 1$ is given by 
\begin{align*}
	\mathbb P_{\mathrm{out}} = 1 - \exp \left\{ - \left( \dfrac{\xi_n}{P_t^* \Omega_n} + \dfrac{\xi_f}{P_t^* \Omega_f}\right) \right\}.
\end{align*}
Assuming $1/(a_f - a_n \theta) > 1/a_n$, i.e., $\xi_n = 1/(1 - a_n - a_n \theta)$, we have 
\begin{align*}
	\dfrac{\partial \mathbb P_{\mathrm{out}}}{\partial a_n} = & \ \exp \left\{ \dfrac{-\theta (\Omega_n + \Omega_f)}{P_t^* (1 - a_n - a_n \theta) \Omega_n \Omega_f} \right\} \\
	& \hspace{2cm}\times \dfrac{(\Omega_n + \Omega_f) (1 + \theta) \theta}{\Omega_n \Omega_f P_t^* (-1 + a_n + a_n \theta)^2}.
\end{align*}
Since $a_n < 1/ (1 + \theta)$, this implies that 
\begin{align*}
	\exp \left\{ \dfrac{-\theta (\Omega_n + \Omega_f)}{P_t^* (1 - a_n - a_n \theta) \Omega_n \Omega_f} \right\}  \neq 0. 
\end{align*}
Therefore, 
\begin{align*}
	& \dfrac{\partial \mathbb P_{\mathrm{out}}}{\partial a_n} = 0 \\
	\implies & \dfrac{(\Omega_n + \Omega_f) (1 + \theta) \theta}{\Omega_n \Omega_f P_t^* (-1 + a_n + a_n \theta)^2} = 0.
\end{align*}
It can easily be noticed that the only feasible solution for the above equation is $a_n = \pm \infty$. 

On the other hand, when $1/(a_f - a_n \theta) < 1/a_n$, i.e., $\xi_n = 1/a_n$, we have 
\begin{align*}
	\dfrac{\partial \mathbb P_{\mathrm{out}}}{\partial a_n} = & - \exp \left\{ - \left( \dfrac{\theta}{\Omega_n P_t^* a_n} + \dfrac{\theta}{\Omega_f P_t^* (1 - a_n - a_n \theta)}\right)\right\} \\
	& \times \left\{\dfrac{\theta}{\Omega_n P_t^* a_n^2} - \dfrac{(1 + \theta) \theta}{\Omega_f P_t^* (1 - a_n - a_n \theta)^2}\right\}.
\end{align*}
Since $a_n < 1/(1 + \theta)$, this implies that 
\begin{align*}
	\exp \left\{ - \left( \dfrac{\theta}{\Omega_n P_t^* a_n} + \dfrac{\theta}{\Omega_f P_t^* (1 - a_n - a_n \theta)}\right)\right\} \neq 0. 
\end{align*}
Therefore, using the constraint $0 < a_n < 1/(1 + \theta)$, we have 
\begin{align*}
	& \dfrac{\partial \mathbb P_{\mathrm{out}}}{\partial a_n} = 0 	\implies a_n^2 (1 + \theta) \{(1 + \theta) \Omega_f -\Omega_n\} \\
	&  \hspace{3cm}- 2 a_n \Omega_f (1 + \theta) +  \Omega_f = 0 \\
	\implies &  a_n = \dfrac{\Omega_f}{(1 + \theta) \Omega_f - \Omega_n} \pm \dfrac{\sqrt{\Omega_n \Omega_f (1 + \theta)}}{(1 + \theta) \left\{(1 + \theta) \Omega_f -\Omega_n \right\}}.
\end{align*}
Define
\begin{align*}
	a_n^{(1)} = & \ \dfrac{\Omega_f}{(1 + \theta) \Omega_f - \Omega_n} + \dfrac{\sqrt{\Omega_n \Omega_f (1 + \theta)}}{(1 + \theta) \left\{(1 + \theta) \Omega_f -\Omega_n \right\}} \\
	= & \ \dfrac{1}{(1 + \theta) - (\Omega_n/\Omega_f)} + \dfrac{\sqrt{\tfrac{\Omega_n}{(1 + \theta) \Omega_f}}}{(1 + \theta) - (\Omega_n/\Omega_f)} \\
	\triangleq & \ \mathscr X_1 + \mathscr X_2. 
\end{align*}
For the case when $1 + \theta > \Omega_n / \Omega_f$, $\mathscr X_1, \mathscr X_2 > 0$ and $\mathscr X_1 > 1/ (1 + \theta)$, leading to $a_n^{(1)} > 1/(1 + \theta)$. On the other hand, if $1 + \theta < \Omega_n / \Omega_f$, $\mathscr X_1 , \mathscr X_2 < 0$, leading to $a_n^{(1)} < 0$. Therefore, the only feasible solution is given by 
\begin{align*}
	a_n^* = a_n^{(2)} = \dfrac{\Omega_f}{(1 + \theta) \Omega_f - \Omega_n} - \dfrac{\sqrt{\Omega_n \Omega_f (1 + \theta)}}{(1 + \theta) \left\{(1 + \theta) \Omega_f -\Omega_n \right\}}.
\end{align*}
This completes the proof.
\section{Proof of Theorem~\ref{Sec_PowIntCSI_TheoremCap}} \label{Sec_PowIntCSI_ProofCap}
Given that $N_n = N_f = 1$, we have $g_n = |h_{n, 1}|^2$ and $g_f = |h_{f, 1}|^2$. The expressions for the PDFs of $g_u$ and $g_{\min}$ are respectively given by 
\begin{align*}
	f_{g_u}(x) = \dfrac{1}{\Omega_u} \exp \left( \dfrac{-x}{\Omega_u}\right), \qquad f_{g_{\min}}(x) = \dfrac{1}{\Omega} \exp \left( \dfrac{-x}{\Omega}\right). 
\end{align*}
Solving the first expectation in~\eqref{Sec_PowIntCSI_Cap_Def2}, we have
\begin{small}
\begin{align}
	& \mathbb E_{g_p, g_n} \left\{ \log_2 [1 + a_n g_n P_t^*(g_p)]\right\} \notag \\
	= & \ \int_{y = 0}^{y = \tfrac{I}{P_{\mathrm{peak}}}}\int_{x = 0}^{x = \infty} \log_2 (1 + a_n P_{\mathrm{peak}} x) f_{g_n}(x) f_{g_p}(y)\mathrm dx \ \mathrm dy \notag \\
& + \int_{y = \tfrac{I}{P_{\mathrm{peak}}}}^{y = \infty} \int_{x = 0}^{x = \infty} \log_2 \left( 1 + a_n I \dfrac{x}{y}\right) f_{g_n}(x) f_{g_p}(y)\mathrm dx \ \mathrm dy  \notag \\
= & \dfrac{1}{\Omega_n \Omega_p} \int_{y = 0}^{y = \tfrac{I}{P_{\mathrm{peak}}}}\int_{x = 0}^{x = \infty} \log_2 (1 + a_n P_{\mathrm{peak}} x) \exp\left( \dfrac{-x}{\Omega_n} \right) \notag \\
& \hspace{5cm}\times \exp\left( \dfrac{-y}{\Omega_p} \right) \mathrm dx \ \mathrm dy \notag \\
& + \dfrac{1}{\Omega_n \Omega_p} \int_{y = \tfrac{I}{P_{\mathrm{peak}}}}^{y =\infty} \int_{x = 0}^{x = \infty} \log_2 \left( 1 + a_n I \dfrac{x}{y}\right) \exp\left( \dfrac{-x}{\Omega_n} \right) \notag \\
& \hspace{5cm}\times \exp\left( \dfrac{-y}{\Omega_p} \right) \mathrm dx \ \mathrm dy \notag \\
= & \dfrac{1}{\Omega_n \Omega_p \ln 2} \left[ \int_{y = 0}^{y = \tfrac{I}{P_{\mathrm{peak}}}} \exp \left( \dfrac{-y}{\Omega_p}\right) \mathrm dy \right] \notag \\ 
& \hspace{1.5cm}\times \left[ \int_{x = 0}^{x = \infty} \ln(1 + a_n P_{\mathrm{peak}} x) \exp \left( \dfrac{-x}{\Omega_n}\right) \mathrm dx \right] \notag \\
& + \dfrac{1}{\Omega_n \Omega_p \ln 2} \int_{y = \tfrac{I}{P_{\mathrm{peak}}}}^{y =\infty} \int_{x = 0}^{x = \infty} \ln \left( 1 + a_n I \dfrac{x}{y}\right) \exp\left( \dfrac{-x}{\Omega_n} \right) \notag \\
& \hspace{5cm}\times \exp\left( \dfrac{-y}{\Omega_p} \right) \mathrm dx \ \mathrm dy \notag \\
= & \ \dfrac{-1}{\ln 2} \left[ 1 - \exp \left( \dfrac{-I}{\Omega_p P_{\mathrm{peak}}}\right)\right] \exp \left( \dfrac{1}{a_n \Omega_n P_{\mathrm{peak}}}\right) \notag \\
& \times \operatorname{Ei} \left( \dfrac{-1}{a_n \Omega_n P_{\mathrm{peak}}} \right) \!+\! \dfrac{1}{\ln 2} \left[ \operatorname{Ei} \left( \!\dfrac{-I}{\Omega_p P_{\mathrm{peak}}}\!\right) -\dfrac{\Omega_p}{\Omega_p - a_n \Omega_n I} \right. \notag \\
&  \times \! \left\{\operatorname{Ei} \!\left(\! \dfrac{-I}{\Omega_p P_{\mathrm{peak}}}\! \right)\! - \! \exp \left(\! \dfrac{\Omega_p - a_n \Omega_n I}{a_n \Omega_n \Omega_p P_{\mathrm{peak}}}\!\right) \operatorname{Ei} \! \left( \!\dfrac{-1}{a_n \Omega_n P_{\mathrm{peak}}}\!\right)\!\right\}  \notag \\
&  + \exp \left( \dfrac{\Omega_p - a_n \Omega_n I}{a_n \Omega_n \Omega_p P_{\mathrm{peak}}}\right) \notag \\
& \hspace{1cm}\left. \times \left\{\operatorname{Shi}\left( \dfrac{1}{a_n \Omega_n P_{\mathrm{peak}}}\right) - \operatorname{Chi}\left( \dfrac{1}{a_n \Omega_n P_{\mathrm{peak}}}\right) 	\right\} \right] \notag \\
\triangleq & \ T(a_n \Omega_n). \label{Sec_PowIntCSI_T1}
\end{align}
\end{small}
Similarly, the analytical expression for the second expectation in~\eqref{Sec_PowIntCSI_Cap_Def2} can be obtained by replacing $a_n \Omega_n$ by $\Omega$, and the analytical expression for the third expectation in~\eqref{Sec_PowIntCSI_Cap_Def2} can be obtained by replacing $\Omega_n$ by $\Omega$ in the preceding equation; this completes the proof.
\section{Proof of Theorem~\ref{Sec_PowIntCSI_TheoremOut}} \label{Sec_PowIntCSI_ProofOut}
Using~\eqref{Sec_PowIntCSI_PtOptimal} and~\eqref{Sec_PowIntCSI_Out_Def}, it follows that 
\begin{small}
\begin{align}
	& \mathbb P_{\mathrm{out}} = 1 - \Pr \left( g_n \geq \dfrac{\xi_n}{P_t^*(g_p)}\right) \Pr \left( g_f \geq \dfrac{\xi_f}{P_t^*(g_p)}\right)  \notag \\
	\!= & 1 - \left\{ \! \underbrace{\Pr \left(\! g_n \!\geq\! \dfrac{\xi_n}{P_{\mathrm{peak}}} \! \right) \Pr \left(\! g_f \!\geq\! \dfrac{\xi_f}{P_{\mathrm{peak}}} \!\right) \Pr \left(\! g_p \!\leq\! \dfrac{I}{P_{\mathrm{peak}}} \!\right)}_{\mathfrak X_1} \right. \notag \\
	& \hspace{1cm}\left. + \underbrace{ \Pr \left( \dfrac{g_n}{g_p} \geq \dfrac{\xi_n}{I}, \dfrac{g_f}{g_p} \geq \dfrac{\xi_f}{I}, g_p > \dfrac{I}{P_{\mathrm{peak}}} \right) }_{\mathfrak X_2}\right\}. \label{Sec_PowIntCSI_Out_Def2}
\end{align}
\end{small}
Solving for $\mathfrak X_1$ yields
\begin{small}
\begin{align}
	& \mathfrak X_1 = \left[ \prod_{u \in \{n, f\}} \int_{\tfrac{\xi_u}{P_{\mathrm{peak}}}}^{\infty} f_{g_u}(x) \mathrm dx\right]  \left[ 1 - \exp \left( \dfrac{-I}{\Omega_p P_{\mathrm{peak}}}\right)\right] \notag \\
	= & \ \left[ 1 - \exp \left( \dfrac{-I}{\Omega_p P_{\mathrm{peak}}}\right)\right] \prod_{u \in \{n, f\}} \!\! \dfrac{1}{\Gamma(N_u) \Omega_u^{N_u}}\int_{\tfrac{\xi_u}{P_{\mathrm{peak}}}}^{\infty} \!\! x^{N_u - 1} \notag \\
	& \hspace{5.5cm} \times \exp \left( \dfrac{-x}{\Omega_u}\right) \mathrm dx \notag \\
	= & \ \left[ 1 - \exp \left( \dfrac{-I}{\Omega_p P_{\mathrm{peak}}}\right)\right] \prod_{u \in \{n, f\}} \dfrac{\Gamma[N_u, \xi_u/(\Omega_u P_{\mathrm{peak}})]}{\Gamma(N_u)}. \label{mathfrakX1}
\end{align}
\end{small}
The integral above is solved using~\cite[eqn.~(3.381-3),~p.~346]{Grad}. Similarly, solving for $\mathfrak X_2$ yields
\begin{small}
\begin{align}
	& \mathfrak X_2 = \Pr \left( \dfrac{g_n}{g_p} \geq \dfrac{\xi_n}{I}, \dfrac{g_f}{g_p} \geq \dfrac{\xi_f}{I}, g_p > \dfrac{I}{P_{\mathrm{peak}}} \right) \notag \\
	= & \ \dfrac{1}{\Omega_p} \int_{\tfrac{I}{P_{\mathrm{peak}}}}^{\infty} \! \exp \left(\! \dfrac{-x}{\Omega_p} \!\right) \Pr \left(\! g_n \!\geq\! \dfrac{\xi_n}{I} x\!\right) \Pr \left(\! g_f \!\geq\! \dfrac{\xi_f}{I} x\!\right) \mathrm dx \notag \\
	= & \ \dfrac{1}{\Omega_p} \int_{\tfrac{I}{P_{\mathrm{peak}}}}^{\infty} \!\!\!\! \exp \left( \!\dfrac{-x}{\Omega_p}\! \right)\! \!\left[\! 1 \!-\! F_{g_n} \left(\! \dfrac{\xi_n}{I} x \!\right)\!\right] \!\!\left[\! 1 \!-\! F_{g_f} \!\left( \! \dfrac{\xi_f}{I} x \!\right)\!\right] \mathrm dx \notag \\
	= & \ \dfrac{1}{\Omega_p} \int_{\tfrac{I}{P_{\mathrm{peak}}}}^{\infty} \exp \left( \dfrac{-x}{\Omega_p} \right) \exp \left( \dfrac{-\xi_n x}{\Omega_n I } \right) \sum_{l = 0}^{N_n - 1} \dfrac{1}{l!} \left( \dfrac{\xi_n x}{\Omega_n I} \right)^{l} \notag \\
	& \hspace{2cm}\times \exp \left( \dfrac{-\xi_f x}{\Omega_f I } \right) \sum_{k = 0}^{N_f - 1} \dfrac{1}{k!} \left( \dfrac{\xi_f x}{\Omega_f I} \right)^{k} \mathrm dx \notag \\
	= & \ \dfrac{1}{\Omega_p} \sum_{k = 0}^{N_f - 1} \sum_{l = 0}^{N_n - 1} \dfrac{\xi_n^l \xi_f^k}{k! l! \Omega_n^l \Omega_f^k I^{k + l}} \int_{\tfrac{I}{P_{\mathrm{peak}}}}^{\infty} x^{k + l} \notag \\
	& \hspace{2.2cm}\times \exp \left[ - \left( \dfrac{1}{\Omega_p} + \dfrac{\xi_n}{\Omega_n I} + \dfrac{\xi_f}{\Omega_n I} \right) x\right] \mathrm dx \notag \\
	= & \ \dfrac{1}{\Omega_p} \sum_{k = 0}^{N_f - 1} \sum_{l = 0}^{N_n - 1} \!\! \Gamma \left[ k + l + 1, \left( \!\dfrac{1}{\Omega_p} \!+\! \dfrac{\xi_n}{\Omega_n I} \!+\! \dfrac{\xi_f}{\Omega_n I} \right) \dfrac{I}{P_{\mathrm{peak}}}\right] \notag \\ 
	& \hspace{0.3cm}\times \dfrac{\xi_n^l \xi_f^k}{k! l! \Omega_n^l \Omega_f^k I^{k + l}} \left( \dfrac{1}{\Omega_p} + \dfrac{\xi_n}{\Omega_n I} + \dfrac{\xi_f}{\Omega_n I} \right)^{-(k + l + 1)}. \label{mathfrakX2}
\end{align}
\end{small}
The integral above is solved using~\cite[eqn.~(3.381-3),~p.~346]{Grad}. Therefore, using~\eqref{Sec_PowIntCSI_Out_Def2}-\eqref{mathfrakX2}, an analytical expression for~\eqref{Sec_PowIntCSI_Out_Def} is given by~\eqref{Sec_PowIntCSI_OutClosed}; this concludes the proof.
\balance
\bibliographystyle{IEEEtran}
\bibliography{NOMA-SS}
\end{document}